\documentclass[12pt]{article}

\usepackage{amsmath,amssymb,amsfonts}
\usepackage{amsthm}
\usepackage{booktabs}
\usepackage{multirow}
\usepackage{array}
\usepackage{url}
\usepackage{graphicx}
\usepackage{xcolor}
\usepackage{float}
\usepackage{placeins}
\usepackage{geometry}
\usepackage{hyperref}
\usepackage{adjustbox}

\newtheorem{theorem}{Theorem}
\newtheorem{lemma}{Lemma}
\newtheorem{proposition}{Proposition}
\newtheorem{definition}{Definition}
\newtheorem{remark}{Remark}
\newtheorem{example}{Example}

\newtheorem{corollary}{Corollary}

\newcommand{\Z}{\mathbb{Z}}

\newcommand{\EJ}{\mathrm{EJ}}
\newcommand{\hex}{\mathcal{H}_t}

\title{Local Fault Repair of Perfect Resource Placements in Eisenstein--Jacobi Networks}

\author{Bader A. Albader\\
\small Department of Computer Science, Faculty of Science, Kuwait University, Kuwait\\
\small \texttt{albader@cs.ku.edu.kw}}

\date{}

\begin{document}

\maketitle

\begin{abstract}
Perfect resource placements in dense Eisenstein--Jacobi networks partition the network into hexagonal radius-$t$ service cells.  The fault-free placement problem is already classified; this paper studies the complementary post-deployment problem of locally repairing such placements after resource failures.  For the dense Eisenstein--Jacobi family generated by $\alpha=n+(n-1)\omega$, with $t=n-1$, we first prove failure-cell locality and candidate locality.  For one failed resource, we prove that one nonfailed replacement cannot cover the failed hexagon and that two replacements always suffice; hence $\rho_{\mathrm{EJ}}(t)=2$ for every $t\ge1$.  Among all minimum-size one-fault repairs, we prove the sharp minimum-overlap formula $\Omega_{\mathrm{EJ}}(t)=t^2$.  The lower bound follows from the three-strip geometry of Eisenstein--Jacobi balls: every two-ball cover of a failed hexagonal cell contains a forced $t\times t$ axial interface region.  We then extend the framework beyond one fault.  For two failed resources, independent repair gives a universal four-replacement upper bound, but unlike the Gaussian case the EJ geometry is not always additive: we give explicit three-replacement constructions for two infinite neighboring displacement families and prove algebraic three-strip constructions for the infinite non-additive neighboring families.  For additive neighboring pairs, we prove two algebraic lower-bound mechanisms: an endpoint-rigidity theorem for pairs with separated opposite axial endpoints, and a diagonal-corridor theorem for the remaining family $D=\pm(g_1+g_2)$.  Thus the closed-form two-fault section separates the non-additive neighboring families from the additive mechanisms without leaving the diagonal corridor as an open case.  For $q$ failed resources, independent canonical repair gives a universal $2q$ upper bound, and this bound is exact whenever failed cells are separated by more than $4t$ in Eisenstein--Jacobi distance.  We also identify and prove infinite dense four-fault and six-fault cluster families: the four-fault cluster has exact repair number four rather than eight, and the six-fault cluster has exact repair number five rather than twelve, showing strong multi-fault subadditivity.  For arbitrary multi-fault repairs, we prove an exact inclusion--exclusion identity for repeated coverage inside the failed region and its low-multiplicity specializations.  A second exact optimization audit over 19,400 multi-fault instances for $2\le t\le12$ and $3\le q\le6$ confirms widespread subadditivity, including $q=6$ instances with saving seven relative to independent repair.  The results show that Eisenstein--Jacobi local repair is not a direct copy of the Gaussian case: the one-fault overlap is quadratic, neighboring two-fault repair can drop from four to three, and dense clustered repairs can reuse replacement balls across several failed hexagonal cells.
\end{abstract}


\section{Introduction}

Resource placement is a basic abstraction for allocating servers, controllers, memory modules, I/O nodes, or monitoring nodes in a large interconnection network.  A radius-$t$ placement is perfect when every network vertex is within distance $t$ of exactly one resource.  Equivalently, the radius-$t$ balls centered at the resource nodes partition the network.

Perfect resource placement in Gaussian and Eisenstein--Jacobi interconnection networks was classified by Flahive and Bose \cite{FlahiveBose2013}.  That result is fault-free: it determines when a perfect placement exists and describes the corresponding lattice of resource nodes.  The present paper starts after deployment.  It assumes that a classified perfect placement is already active and asks what must be done after a resource node fails.

The operational objective is local repair.  If a resource $r$ fails, the vertices outside its former service cell remain served by the surviving resources.  Thus the post-fault task is not to recompute a new global perfect placement, but to activate a small set of nearby replacement resources that covers the failed cell.  This creates a finite geometric optimization problem: minimize the number of replacements and, among minimum-size repairs, minimize the number of vertices that are served more than once inside the failed cell.

From a systems perspective, constant-size repair rules are useful because a controller does not need to solve a global placement problem after a resource fault.  Once the dense EJ placement is deployed, the controller needs only the failed resource coordinate, the radius $t$, and the local displacement class of any clustered fault pattern.  The replacement centers are then obtained from closed-form coordinate rules or from a small local pattern library.  This supports fast reconfiguration in large interconnection networks where global remapping would be expensive, disruptive, or undesirable during transient resource, controller, memory-module, or service-node failures.

This paper develops the Eisenstein--Jacobi analogue of the Gaussian local-repair framework.  The separation is necessary because the two geometries have different extremal structure.  In the Gaussian case, Lee balls become parity-constrained squares after a coordinate rotation, and the sharp one-fault overlap is linear in $t$.  In the Eisenstein--Jacobi case, balls are discrete hexagons described by three axial strip constraints, and the sharp one-fault overlap is quadratic.

Fault tolerance in graph-based networks has also been studied through complementary robustness parameters, including fault-tolerant metric dimension \cite{PrabhuKlavzar2022}, diagnosability of interconnection networks \cite{ChengMaoQiuShen2022}, and fault-tolerant domination variants such as power domination \cite{GirishSomasundaram2024}.  The broader domination literature contains foundational and survey work on perfect and independent domination \cite{LivingstonStout1990,Klostermeyer2015,GoddardHenning2013}.  Resource placement in tori and related interconnection networks was studied earlier in \cite{BaeBose1997,RamanathanChalasani1995,TzengFeng1996}.  Eisenstein--Jacobi and hexagonal network models have been developed in \cite{MartinezEJ2008,FlahiveBose2010,ThomsonZhou2012,AlbaderBoseFlahive2012}.  The present paper is different in focus: the original resource set is already a perfect dominating set, and the objective is to determine the exact local replacement cost after a deployed resource fails.

The contributions are as follows.
\begin{itemize}
    \item We formulate local fault repair for perfect Eisenstein--Jacobi resource placements as a post-deployment domination-recovery problem.
    \item We prove failure-cell locality and candidate locality: after one resource fails, only its former hexagonal service cell can become uncovered, and replacements farther than $2t$ from the failed resource cannot contribute.
    \item We prove the exact one-fault replacement number $\rho_{\mathrm{EJ}}(t)=2$ for all $t\ge1$.
    \item We prove the exact one-fault minimum-overlap formula $\Omega_{\mathrm{EJ}}(t)=t^2$ among minimum-size repairs.
    \item We identify a canonical family of $3t$ two-center repairs attaining the optimum overlap, and we verify by exhaustive enumeration that these are exactly the optimal pairs for $1\le t\le20$.
    \item We prove a universal four-replacement upper bound for two failed resources and show algebraically that dense EJ repair is non-additive for specific neighboring displacement families: three replacements are both sufficient and necessary for those infinite families.
    \item We add endpoint-rigidity and diagonal-corridor theorems for additive two-fault pairs.  Together these close the symbolic additive mechanisms needed after the explicit non-additive families are separated.
    \item We prove a universal $2q$ upper bound for $q$ failed resources by independent canonical repair and prove exact $2q$ additivity for pairwise $4t$-separated failed cells.
    \item We identify dense four-fault and six-fault cluster families with exact repair numbers four and five, respectively, giving proved savings of four and seven replacements relative to independent repair.
    \item We prove an exact multi-failure overlap identity using replacement multiplicities inside the failed region, together with explicit low-multiplicity specializations such as $P_2-P_3+P_4$.
    \item We include reproducible one-fault and multi-fault audits analogous to the Gaussian study, recomputing coverage and overlap directly from lattice geometry.
\end{itemize}

Table~\ref{tab:main-results} summarizes the formal results.

\begin{table}[H]
\caption{Main results and proof locations.}
\label{tab:main-results}
\centering
\scriptsize
\begin{adjustbox}{max width=\textwidth}
\begin{tabular}{p{0.24\textwidth}p{0.30\textwidth}p{0.28\textwidth}}
\toprule
Setting & Result & Proof method \\
\midrule
One failed resource & Only the failed hexagonal cell can become uncovered; candidates outside radius $2t$ are irrelevant & Perfectness and triangle inequality \\
One failed resource & $\rho_{\mathrm{EJ}}(t)=2$ for all $t\ge1$ & Six extreme vertices and explicit two-center cover \\
One failed resource, minimum overlap & $\Omega_{\mathrm{EJ}}(t)=t^2$ & Three-strip slice lower bound and tight axial-interface construction \\
Canonical one-fault repairs & $3t$ canonical optimum-overlap pairs & Three axial orientations and $t$ interface positions per orientation; exact audit for $t\le20$ \\
Two failed resources & $\rho^{(2)}_{\mathrm{EJ}}(t,F)\le4$ always; non-additive $K=3$ neighboring families exist & Independent repair upper bound and explicit three-ball constructions \\
Two-fault validation & $78$ non-additive and $754$ additive neighboring cases for $2\le t\le20$ & Bounded exact optimizer audit; not used as a universal theorem \\
$q$ failed resources & $\rho^{(q)}_{\mathrm{EJ}}(t,F)\le2q$ & Independent translated canonical repairs \\
Pairwise separated failures & $\rho^{(q)}_{\mathrm{EJ}}(t,F)=2q$ if all failed centers are more than $4t$ apart & Candidate-intersection locality plus one-cell lower bound \\
General multi-fault overlap & $O(R)=\sum_{j\ge2}(-1)^jP_j(R)$ & Vertexwise binomial inclusion--exclusion over multiplicities \\
\bottomrule
\end{tabular}
\end{adjustbox}
\end{table}

\section{Eisenstein--Jacobi Preliminaries}

Let
\begin{equation}
    \omega=\frac{-1+i\sqrt{3}}{2}.
\end{equation}
The Eisenstein integer lattice is
\begin{equation}
    \Z[\omega]=\{x+y\omega:x,y\in\Z\}.
\end{equation}
We use axial coordinates and identify $x+y\omega$ with the pair $(x,y)$.  Two vertices are adjacent when their difference is one of the six unit Eisenstein directions
\begin{equation}
    \pm(1,0),\qquad \pm(0,1),\qquad \pm(1,-1).
\end{equation}
The graph distance from the origin is
\begin{equation}
    d_{\EJ}((0,0),(x,y))=\max\{|x|,|y|,|x+y|\}.
    \label{eq:ej-distance}
\end{equation}
By translation,
\begin{equation}
    d_{\EJ}((a,b),(x,y))=
    \max\{|x-a|,|y-b|,|x+y-a-b|\}.
    \label{eq:ej-distance-translated}
\end{equation}

The radius-$t$ ball centered at $c$ is
\begin{equation}
    B_t(c)=\{u:d_{\EJ}(u,c)\le t\}.
\end{equation}
The central ball is the discrete hexagon
\begin{equation}
    \hex=B_t(0)=\{(x,y): |x|\le t,\ |y|\le t,\ |x+y|\le t\}.
    \label{eq:central-ej-ball}
\end{equation}
It has
\begin{equation}
    |B_t(0)|=3t^2+3t+1
    \label{eq:ej-ball-size}
\end{equation}
vertices.  Figure~\ref{fig:ej-three-strip-ball} shows the three-strip description of this hexagonal ball and labels the six extreme vertices used later in the repair lower bounds.

\begin{figure}[H]
\centering
\includegraphics[width=0.55\linewidth]{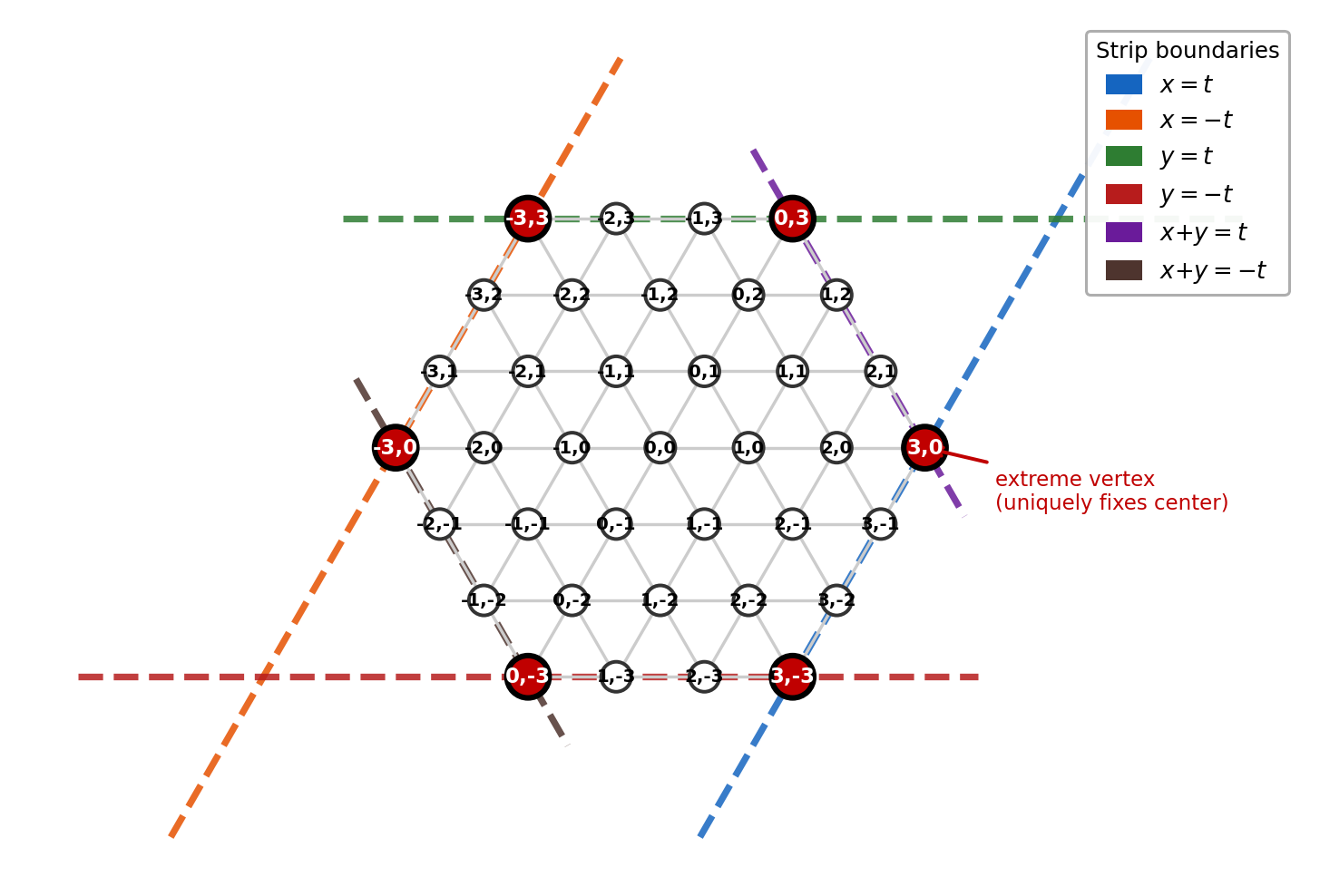}
\caption{EJ radius-$t$ ball as the intersection of three axial strips for $t=3$.  The discrete hexagon $B_t(0)$ is defined by $|x|\le t$, $|y|\le t$, and $|x+y|\le t$; its six extreme vertices determine the ball center and drive the one-fault impossibility proof.}
\label{fig:ej-three-strip-ball}
\end{figure}

For the dense Eisenstein--Jacobi family generated by
\begin{equation}
    \alpha=n+(n-1)\omega,
    \qquad t=n-1,
\end{equation}
the network order is
\begin{equation}
    N=3n^2-3n+1=3t^2+3t+1.
\end{equation}
The integer label associated with the axial coordinate $(x,y)$ is
\begin{equation}
    \phi(x+y\omega)
    \equiv (n-1)x-ny \pmod N.
    \label{eq:ej-labeling}
\end{equation}
The labeling is useful for implementation and for matching numbered EJ-network diagrams.  The proofs below are stated in axial coordinates, where the three-strip structure of the balls is explicit.

\begin{definition}[Perfect EJ $t$-placement]
A resource set $S$ is a perfect EJ $t$-placement if every vertex of the network is within Eisenstein--Jacobi distance $t$ of exactly one resource node in $S$.
\end{definition}

\begin{theorem}[EJ resource-placement classification \cite{FlahiveBose2013}]
Eisenstein--Jacobi networks admit perfect $t$-dominating placements exactly in the canonical divisibility cases.  In the dense case generated by $\alpha=n+(n-1)\omega$, the corresponding placement is a translate of the ideal generated by $\alpha$, with service radius $t=n-1$.
\end{theorem}

\begin{proposition}[Associate and conjugate generator symmetry]
Let
\begin{equation}
    \alpha_0=n+(n-1)\omega.
\end{equation}
Every dense EJ generator obtained from $\alpha_0$ by multiplication by an Eisenstein unit in
\begin{equation}
    \mathcal U=\{1,-1,\omega,-\omega,\omega^2,-\omega^2\}
\end{equation}
or by first conjugating and then multiplying by a unit has the same local-repair parameters as $\alpha_0$.  In particular, the replacement numbers, overlap values, and dense-cluster identities proved below are invariant under all associate and conjugate companion generator choices.
\end{proposition}

\begin{proof}
Write an Eisenstein--Jacobi element as a coordinate pair $(a,b)$ meaning $a+b\omega$.  Multiplication by the six units acts as
\begin{center}
\begin{tabular}{c|c}
unit $u$ & coordinate image of $(a,b)$ for $u(a+b\omega)$ \\
\hline
$1$ & $(a,b)$ \\
$-1$ & $(-a,-b)$ \\
$\omega$ & $(-b,a-b)$ \\
$-\omega$ & $(b,-a+b)$ \\
$\omega^2$ & $(-a+b,-a)$ \\
$-\omega^2$ & $(a-b,a)$
\end{tabular}
\end{center}
Conjugation sends $(a,b)$ to $(a-b,-b)$ because $\overline{\omega}=\omega^2=-1-\omega$.  Applying these maps to $\alpha_0=(n,n-1)$ gives the associate orbit
\begin{equation}
\begin{split}
&(n,n-1),\;(-n,1-n),\;(1-n,1),\;(n-1,-1),\\
&(-1,-n),\;(1,n),
\end{split}
\end{equation}
and applying them to the conjugate generator $\overline{\alpha_0}=(1,1-n)$ gives
\begin{equation}
\begin{split}
&(1,1-n),\;(-1,n-1),\;(n-1,n),\;(1-n,-n),\\
&(-n,-1),\;(n,1).
\end{split}
\end{equation}
These maps are automorphisms of the triangular lattice: they permute the six unit directions $\pm(1,0)$, $\pm(0,1)$, and $\pm(1,-1)$ and therefore preserve the distance formula \eqref{eq:ej-distance}.  They map radius-$t$ EJ balls to radius-$t$ EJ balls, map perfect placement cells to congruent perfect placement cells, and preserve coverage multiplicities.  Since the local repair model is defined only through these metric balls and their incidences, every local-repair statement transfers unchanged to all listed associate or conjugate dense generators.
\end{proof}

\section{Fault-Recovery Model}

Let $S$ be a perfect EJ $t$-placement and let $r\in S$ fail.  A repair set $R\subseteq V\setminus S$ is valid if
\begin{equation}
    B_t(r)\subseteq\bigcup_{q\in R}B_t(q).
\end{equation}
The one-fault repair number is
\begin{equation}
    \rho_{\mathrm{EJ}}(t)=\min\{|R|:R\text{ is a valid repair for }r\}.
\end{equation}
By translation invariance, $\rho_{\mathrm{EJ}}(t)$ does not depend on $r$.

The translated canonical repair algorithm places one two-center repair pair at the resource coordinate $r$:
\begin{equation}
    R^*(r)=\{r-(1,0),\ r+t(1,0)\}.
    \label{eq:translated-algorithm}
\end{equation}
By the two-center cover lemma proved below, $R^*(r)$ is a valid repair for every $r$.

Among minimum-size repair sets, we also measure the unavoidable repeated coverage inside the failed cell.  For a two-replacement repair $R=\{q_1,q_2\}$, define
\begin{equation}
    \operatorname{ov}(R)
    =|B_t(r)\cap B_t(q_1)\cap B_t(q_2)|.
\end{equation}
The minimum-overlap value is
\begin{equation}
    \Omega_{\mathrm{EJ}}(t)=
    \min\{\operatorname{ov}(R): |R|=\rho_{\mathrm{EJ}}(t),\ R\text{ repairs }r\}.
\end{equation}

\section{Locality Theorems}

The locality argument is independent of the detailed shape of the EJ ball.  It uses only perfectness of the original placement and the metric triangle inequality.

\begin{theorem}[Failure-cell locality]
Let $S$ be a perfect EJ $t$-placement and let $r\in S$ fail.  Then the only vertices that can become uncovered are the vertices in $B_t(r)$.
\end{theorem}

\begin{proof}
Since $S$ is a perfect $t$-placement, every vertex was covered by exactly one resource.  If a vertex $u$ is not in $B_t(r)$, then $r$ was not the resource covering $u$.  Therefore the unique resource that covered $u$ remains active after $r$ is removed.  Only vertices originally covered by $r$, namely the vertices of $B_t(r)$, can become uncovered.
\end{proof}

\begin{theorem}[Candidate locality]
If a candidate replacement $q$ satisfies $d_{\EJ}(q,r)>2t$, then $B_t(q)\cap B_t(r)=\emptyset$.  Hence $q$ cannot contribute to repairing the failed cell.
\end{theorem}

\begin{proof}
If $u\in B_t(q)\cap B_t(r)$, then
\begin{equation}
    d_{\EJ}(q,r)\le d_{\EJ}(q,u)+d_{\EJ}(u,r)\le2t,
\end{equation}
contradicting $d_{\EJ}(q,r)>2t$.
\end{proof}

\section{Exact One-Fault Repair Number}

By translation invariance, assume that the failed resource is the origin.  Thus the failed cell is the central hexagon $\hex=B_t(0)$.

\begin{figure}[H]
\centering
\includegraphics[width=0.5\linewidth]{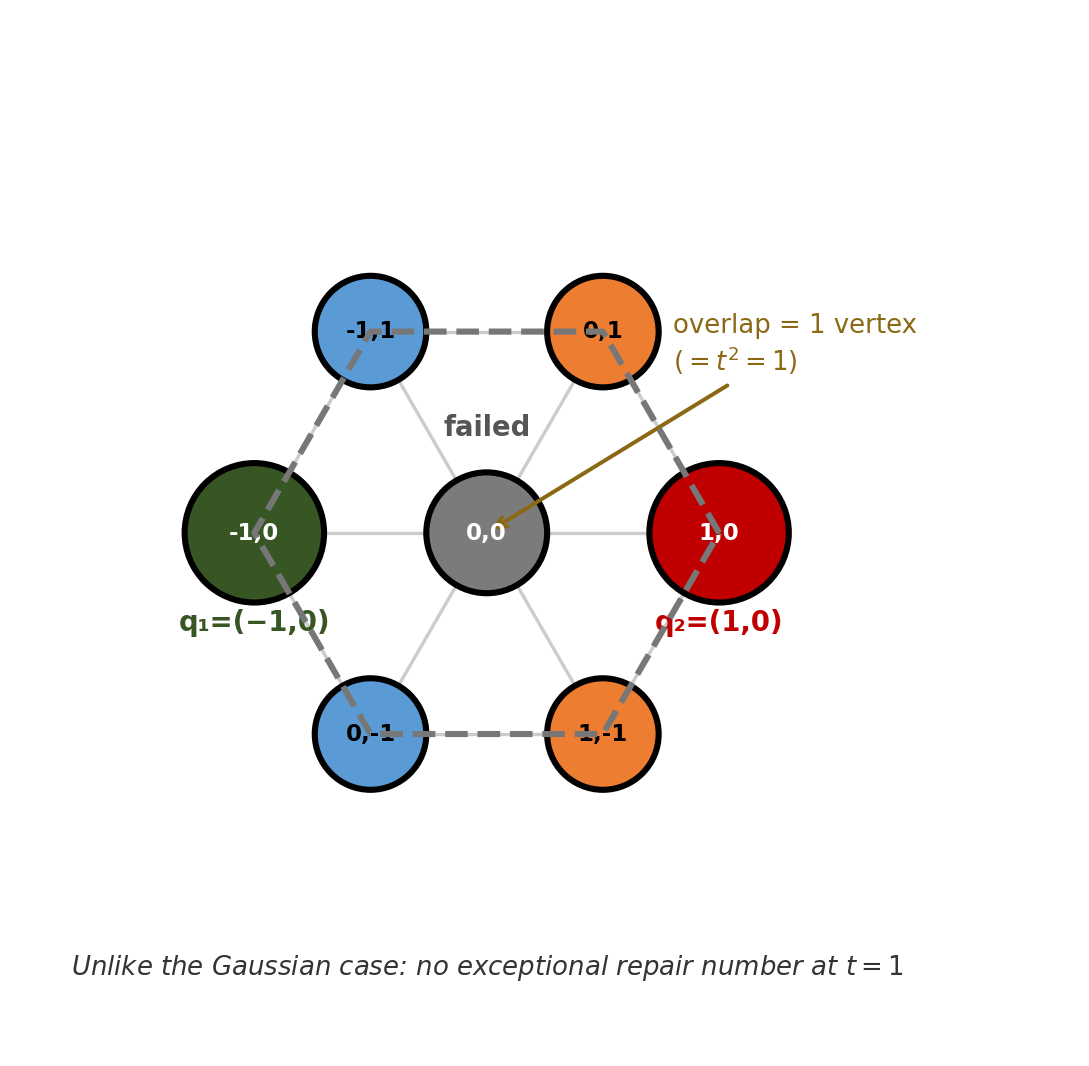}
\caption{The small case $t=1$.  Unlike the Gaussian case, EJ has no exceptional repair number: two replacement balls cover the seven-vertex failed hexagon, and their overlap contains exactly one vertex.}
\label{fig:ej-t1}
\end{figure}

\begin{lemma}[A single nonfailed replacement is impossible]
For every $t\ge1$, no vertex $q\ne(0,0)$ satisfies
\begin{equation}
    B_t(0)\subseteq B_t(q).
\end{equation}
\end{lemma}

\begin{proof}
The two balls have the same finite cardinality.  Therefore containment would imply $B_t(0)=B_t(q)$.  The center of an EJ ball is determined uniquely by its six extreme vertices
\begin{equation}
    (t,0),\quad (0,t),\quad (-t,t),\quad (-t,0),\quad (0,-t),\quad (t,-t).
\end{equation}
Equivalently, every pair of opposite extreme vertices has midpoint equal to the center.  To see this directly, the opposite pair $(t,0)$ and $(-t,0)$ forces the $x$-midline of the ball to be $x=0$; the opposite pair $(0,t)$ and $(0,-t)$ forces the $y$-midline to be $y=0$; and the pair $(-t,t)$ and $(t,-t)$ forces the $(x+y)$-midline to be $x+y=0$.  These three midlines meet only at $(0,0)$.  Equality of the two balls therefore forces $q=(0,0)$, which is unavailable because the resource at the origin has failed.
\end{proof}

\begin{lemma}[Two-center EJ cover]
For every $t\ge1$,
\begin{equation}
    B_t(0)\subseteq B_t((-1,0))\cup B_t((t,0)).
    \label{eq:two-center-cover}
\end{equation}
\end{lemma}

\begin{proof}
Let $(x,y)\in B_t(0)$, so
\begin{equation}
    |x|\le t,\qquad |y|\le t,\qquad |x+y|\le t.
    \label{eq:central-ineqs}
\end{equation}
If $x\le t-1$ and $x+y\le t-1$, then
\begin{equation}
    |x+1|\le t,
    \qquad |y|\le t,
    \qquad |x+y+1|\le t,
\end{equation}
because $x\ge -t$ and $x+y\ge -t$.  Thus $(x,y)\in B_t((-1,0))$.

It remains to consider vertices for which $x=t$ or $x+y=t$.  If $x=t$, then \eqref{eq:central-ineqs} gives $-t\le y\le0$.  Hence
\begin{equation}
    |x-t|=0,
    \qquad |y|\le t,
    \qquad |x+y-t|=|y|\le t,
\end{equation}
so $(x,y)\in B_t((t,0))$.  If $x+y=t$, then \eqref{eq:central-ineqs} gives $0\le x\le t$, and therefore
\begin{equation}
    |x-t|\le t,
    \qquad |y|=|t-x|\le t,
    \qquad |x+y-t|=0.
\end{equation}
Thus $(x,y)\in B_t((t,0))$.  The cases exhaust $B_t(0)$.
\end{proof}

\begin{figure}[H]
\centering
\includegraphics[width=0.6\linewidth]{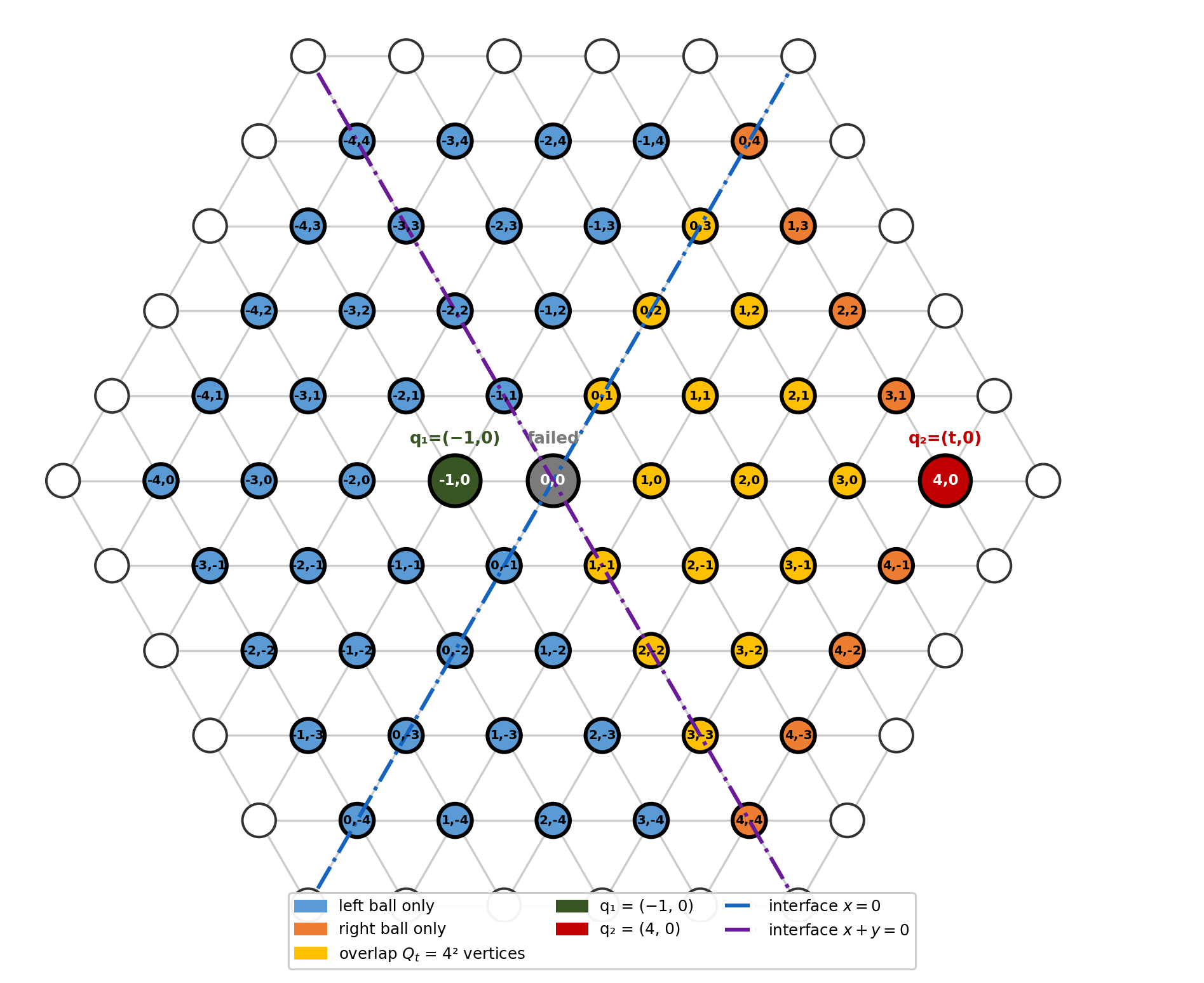}
\caption{Explicit two-center one-fault repair for $t=4$.  Blue vertices are covered only by the left replacement, orange vertices only by the right replacement, and yellow vertices form the overlap region $Q_t$.}
\label{fig:ej-two-center-cover}
\end{figure}

\begin{example}[Worked $t=4$ one-fault repair]
In Fig.~\ref{fig:ej-two-center-cover}, the failed cell is $B_4(0)$ and the repair centers are $(-1,0)$ and $(4,0)$.  The two balls split the hexagon along the axial interface $x=0$ and $x+y=0$; their common part is the highlighted region $Q_4$, containing $4^2=16$ vertices.  This example is the concrete instance of Lemma~2 and of the sharp overlap formula proved later.
\end{example}

\begin{lemma}[Explicit $t=1$ verification]
For $t=1$, the failed EJ cell has seven vertices, two replacements are sufficient, and one replacement is impossible.  Hence the general formula $\rho_{\mathrm{EJ}}(1)=2$ has no small-radius exception.
\end{lemma}

\begin{figure}[H]
\centering
\includegraphics[width=0.65\linewidth]{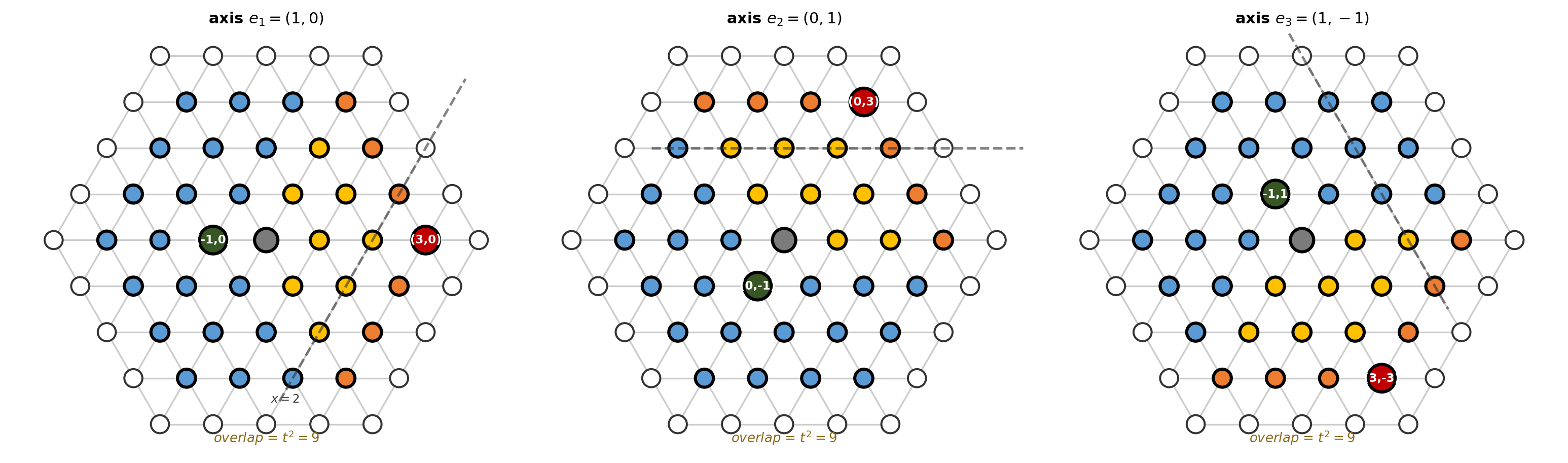}
\caption{Three axial orientations of the canonical optimum repair family.  For each of the three EJ axes and each interface position $a=1,\ldots,t$, the pair $\{-ae_i,(t+1-a)e_i\}$ covers the failed cell and attains overlap $t^2$.}
\label{fig:ej-canonical-family}
\end{figure}

\begin{proof}
The failed cell is
\begin{equation}
B_1(0)=\{(0,0),(1,0),(0,1),(-1,1),(-1,0),(0,-1),(1,-1)\}.
\end{equation}
The pair $\{(-1,0),(1,0)\}$ is the specialization of Lemma~2 and covers all seven vertices.  Lemma~1 rules out one replacement.  Therefore $\rho_{\mathrm{EJ}}(1)=2$.
\end{proof}

\begin{example}[Worked $t=1$ repair]
For $t=1$, the repair pair $\{(-1,0),(1,0)\}$ covers the seven vertices of $B_1(0)$ shown in Fig.~\ref{fig:ej-t1}.  The only repeated vertex is the origin, so the overlap value is $1=t^2$.  Thus the small case agrees with both main one-fault formulas, $\rho_{\mathrm{EJ}}(1)=2$ and $\Omega_{\mathrm{EJ}}(1)=1$.
\end{example}

\begin{theorem}[Exact one-fault EJ repair number]
For every $t\ge1$,
\begin{equation}
    \rho_{\mathrm{EJ}}(t)=2.
\end{equation}
\end{theorem}

\begin{proof}
The previous impossibility lemma rules out one replacement.  The two-center cover in Lemma~2 gives a valid repair with two replacements.
\end{proof}

\section{Canonical Optimal Repair Family}

The explicit pair in Lemma~2 is one member of a larger axial family.  Let the three unoriented EJ axes be represented by
\begin{equation}
    e_1=(1,0),\qquad e_2=(0,1),\qquad e_3=(1,-1).
\end{equation}
For each $a\in\{1,\ldots,t\}$ and each axis direction $e_i$, define
\begin{equation}
    R_{i,a}=\{-ae_i,\ (t+1-a)e_i\}.
    \label{eq:canonical-repair-family}
\end{equation}
There are $3t$ such unordered pairs.  Figure~\ref{fig:ej-canonical-family} illustrates the three axial orientations and the way the interface position changes with $a$.

\begin{proposition}[Canonical optimal pairs]
Every pair $R_{i,a}$ in \eqref{eq:canonical-repair-family} covers $B_t(0)$ and has overlap exactly $t^2$ inside $B_t(0)$.
\end{proposition}

\begin{proof}
It suffices to prove the claim for $e_1=(1,0)$, since the EJ hexagon is invariant under the dihedral symmetries permuting the three axes.  Thus consider
\begin{equation}
    R_{1,a}=\{(-a,0),(t+1-a,0)\},\qquad 1\le a\le t.
\end{equation}
The same argument as Lemma~2 splits the failed hexagon between the two supporting boundaries $x=t-a+1$ and $x+y=t-a+1$.  The first ball covers the vertices satisfying
\begin{equation}
    x\le t-a,
    \qquad x+y\le t-a,
\end{equation}
and the second ball covers all remaining failed-cell vertices with $x\ge t-a+1$ or $x+y\ge t-a+1$.  Therefore the union covers $B_t(0)$.

The common part inside the failed cell is
\begin{equation}
    Q_{t,a}=\{(x,y): 1-a\le x\le t-a,
    \ -x\le y\le t-1-x\}.
    \label{eq:shifted-overlap-region}
\end{equation}
For each of the $t$ admissible $x$-values, the interval for $y$ contains exactly $t$ lattice points.  Thus $|Q_{t,a}|=t^2$.  Symmetry gives the same result for $e_2$ and $e_3$.
\end{proof}

\section{Exact Minimum Overlap for One Fault}

The repair number alone does not measure efficiency.  A two-center repair may cover the failed hexagon, but it necessarily creates repeated coverage inside the failed cell.  We now determine the minimum possible repeated coverage.

For the explicit repair $R^*=\{(-1,0),(t,0)\}$, the overlap inside $B_t(0)$ is the axial rectangle
\begin{equation}
    Q_t=\{(x,y):0\le x\le t-1,\ -x\le y\le t-1-x\}.
    \label{eq:ej-overlap-rectangle}
\end{equation}
For each $x\in\{0,1,\ldots,t-1\}$ there are exactly $t$ admissible values of $y$, so $|Q_t|=t^2$.  Figure~\ref{fig:ej-forced-overlap} shows these $t$ parallel slices explicitly; it is the geometric picture behind the lower-bound proof.

\begin{figure}[H]
\centering
\includegraphics[width=0.6\linewidth]{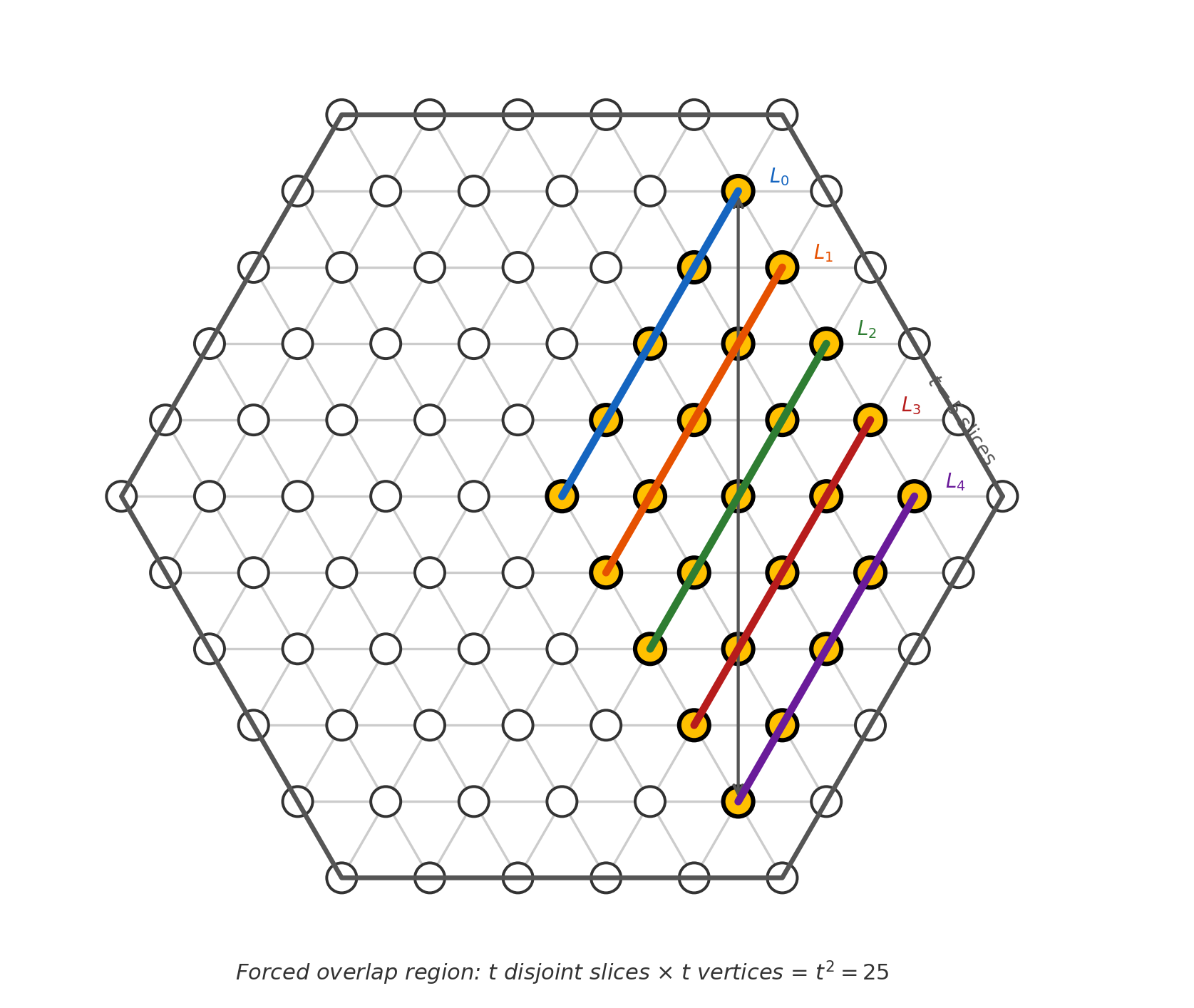}
\caption{Forced axial-interface overlap.  Any two-ball cover has a transition region containing $t$ disjoint axial slices, each with $t$ vertices.  This is the geometric source of the sharp quadratic overlap $\Omega_{\mathrm{EJ}}(t)=t^2$.}
\label{fig:ej-forced-overlap}
\end{figure}

\begin{lemma}[Overlap attained by the explicit repair]
The repair set
\begin{equation}
    R^*=\{(-1,0),(t,0)\}
\end{equation}
has overlap
\begin{equation}
    |B_t(0)\cap B_t((-1,0))\cap B_t((t,0))|=t^2.
\end{equation}
\end{lemma}

\begin{proof}
A vertex $(x,y)\in B_t(0)$ lies in both replacement balls if and only if
\begin{align}
    |x+1|&\le t, & |y|&\le t, & |x+y+1|&\le t, \\
    |x-t|&\le t, & |y|&\le t, & |x+y-t|&\le t.
\end{align}
Together with the central inequalities, these conditions reduce to \eqref{eq:ej-overlap-rectangle}.  Counting gives $t$ choices of $x$ and $t$ choices of $y$ for each $x$, hence $t^2$ vertices.
\end{proof}

\begin{lemma}[Forced axial-interface overlap]
Let $A$ and $B$ be two nonzero EJ vertices such that
\begin{equation}
    B_t(0)\subseteq B_t(A)\cup B_t(B).
    \label{eq:arbitrary-two-cover}
\end{equation}
Then
\begin{equation}
    |B_t(0)\cap B_t(A)\cap B_t(B)|\ge t^2.
\end{equation}
\end{lemma}

\begin{proof}
Write the central hexagon as the intersection of the three strips
\begin{equation}
    -t\le x\le t,
    \qquad -t\le y\le t,
    \qquad -t\le x+y\le t.
    \label{eq:three-strips}
\end{equation}
A translated EJ ball is obtained by shifting the same three strips.  Hence, on every line parallel to one of the three axial directions, its intersection with that line is an interval of consecutive lattice vertices.

Consider the six extreme vertices of $B_t(0)$ in cyclic order, grouped into the three opposite pairs.  A translated radius-$t$ hexagon distinct from $B_t(0)$ cannot contain an opposite pair of these extremes: such a pair fixes the corresponding strip midline, and the three opposite pairs together fix the center at the origin.  Hence each replacement ball can contain at most one extreme from each opposite pair, and therefore at most three extreme vertices.  Since the two replacement balls together cover all six extremes, each ball contains exactly three extremes, one from each opposite pair.  The boundary intersection of a translated hexagon with the central hexagon is connected in cyclic order: on each of the six boundary sides, the three strip inequalities defining the translated ball reduce to consecutive integer intervals, and adjacent nonempty side intervals meet only at shared boundary vertices.  Thus the split is a complementary $3+3$ split.  By a dihedral symmetry of the EJ hexagon and, if necessary, by interchanging $A$ and $B$, the unique transition between the two triples may be placed in the canonical sector bounded by the supporting lines $x=0$ and $x+y=t-1$.  Thus only this canonical $3+3$ transition needs to be analyzed.

For $i=0,1,\ldots,t-1$, define the axial slice
\begin{equation}
    L_i=\{(i,y):-i\le y\le t-1-i\}.
    \label{eq:overlap-slices}
\end{equation}
Each $L_i$ has exactly $t$ vertices, and the slices are pairwise disjoint.  We now spell out the endpoint calculation.  For a center $C=(c_1,c_2)$, the intersection of $B_t(C)$ with the vertical line $x=i$ is the integer interval
\begin{equation}
\begin{split}
I_C(i)=&[c_2-t,c_2+t]\\
&\cap[c_1+c_2-i-t,c_1+c_2-i+t],
\end{split}
\label{eq:vertical-ball-interval}
\end{equation}
provided $|i-c_1|\le t$, and is empty otherwise.  The slice $L_i$ itself is the interval
\begin{equation}
    J_i=[-i,t-1-i].
\end{equation}
We now derive the endpoint ordering explicitly.  The normalized split places the lower transition side on the supporting line $x=0$ and the upper transition side on $x+y=t-1$.  Thus, for each $0\le i\le t-1$, the lower endpoint
\begin{equation}
    \ell_i=(i,-i)\qquad (x=i,\ x+y=0)
\end{equation}
is the first vertex of the transition slice met from the lower arc, while the upper endpoint
\begin{equation}
    u_i=(i,t-1-i)\qquad (x=i,\ x+y=t-1)
\end{equation}
is the last vertex met before the upper arc.  In the normalized order, $B_t(B)$ contains the lower-side transition boundary and $B_t(A)$ contains the upper-side transition boundary.  Hence
\begin{equation}
    -i\in I_B(i),
    \qquad t-1-i\in I_A(i).
    \label{eq:transition-endpoints-contained}
\end{equation}
The remaining endpoint inequalities follow from the shifted strip equations, not from a separate geometric assertion.  Write $A=(a_1,a_2)$ and $B=(b_1,b_2)$.  Since $u_i\in B_t(A)$, the two interval inequalities in \eqref{eq:vertical-ball-interval} give
\begin{equation}
    a_2-t\le t-1-i\le a_2+t,
    \qquad
    a_1+a_2-i-t\le t-1-i\le a_1+a_2-i+t .
    \label{eq:A-upper-contained}
\end{equation}
In the normalized $3+3$ transition, the $A$-side triple occupies the opposite side of the same vertical slice, so the supporting line $x=0$ also places the lower endpoint of the slice in the $A$-side interval.  Thus the shifted strips of $B_t(A)$ give the pair of inequalities
\begin{equation}
    a_2-t\le -i,
    \qquad
    a_1+a_2-i-t\le -i .
    \label{eq:A-left-endpoint}
\end{equation}
Equations \eqref{eq:A-upper-contained} and \eqref{eq:A-left-endpoint} say exactly that
\begin{equation}
    [-i,t-1-i]\subseteq I_A(i).
    \label{eq:Ji-in-A}
\end{equation}
Similarly, since $\ell_i\in B_t(B)$ and the opposite transition side is supported by $x+y=t-1$, substituting $\ell_i$ and $u_i$ into the shifted $y$- and $(x+y)$-strip inequalities gives
\begin{equation}
    b_2-t\le -i\le b_2+t,
    \qquad
    b_1+b_2-i-t\le -i\le b_1+b_2-i+t,
    \label{eq:B-lower-contained}
\end{equation}
plus the right-endpoint inequalities
\begin{equation}
    b_2+t\ge t-1-i,
    \qquad
    b_1+b_2-i+t\ge t-1-i .
    \label{eq:B-right-endpoint}
\end{equation}
Therefore
\begin{equation}
    [-i,t-1-i]\subseteq I_B(i).
\end{equation}
Since $J_i=[-i,t-1-i]\subseteq I_A(i)\cap I_B(i)$, every vertex of $L_i$ lies in both $B_t(A)$ and $B_t(B)$.  Summing over $i=0,\ldots,t-1$ gives at least $t\cdot t=t^2$ vertices in the triple overlap.
\end{proof}

\begin{theorem}[Exact minimum overlap]
For every $t\ge1$,
\begin{equation}
    \Omega_{\mathrm{EJ}}(t)=t^2.
\end{equation}
\end{theorem}

\begin{proof}
The forced overlap lemma gives the lower bound.  The explicit repair $R^*$ attains $t^2$.
\end{proof}

\begin{corollary}[Translated overlap lower bound]
For every resource center $r$ and every valid two-center repair of $B_t(r)$,
\begin{equation}
    \operatorname{ov}(R)\ge t^2.
\end{equation}
\end{corollary}

\begin{proof}
Translate Proposition~2 by $r$.
\end{proof}

\section{Two-Fault Repair: Additivity and Non-Additivity}

The one-fault theorem immediately gives a four-replacement upper bound for two failed resources.  Let
\begin{equation}
    F=\{r_1,r_2\}\subseteq S
\end{equation}
be the failed pair and let
\begin{equation}
    U(F)=B_t(r_1)\cup B_t(r_2)
    \label{eq:failed-region-two}
\end{equation}
be the two-cell failed region.  A repair set $R$ is valid if
\begin{equation}
    U(F)\subseteq\bigcup_{z\in R}B_t(z).
\end{equation}
For a fixed failed pair $F$, write $\rho_{\mathrm{EJ}}^{(2)}(t,F)$ for the minimum number of replacements needed to cover $U(F)$.

\begin{theorem}[Two-fault upper bound]
For every pair of failed resource centers $F=\{r_1,r_2\}$,
\begin{equation}
    \rho_{\mathrm{EJ}}^{(2)}(t,F)\le4.
\end{equation}
\end{theorem}

\begin{proof}
For each failed resource $r_i$, choose one translated canonical repair pair from \eqref{eq:translated-algorithm}; for instance,
\begin{equation}
    R_i=\{r_i-(1,0),\ r_i+t(1,0)\}.
\end{equation}
By Corollary~1, $R_i$ covers $B_t(r_i)$.  Hence $R_1\cup R_2$ covers $B_t(r_1)\cup B_t(r_2)$ and uses at most four replacement centers.
\end{proof}

In the Gaussian local-repair paper, the corresponding two-fault value is always additive for $t\ge2$.  The EJ case is different.  Neighboring EJ hexagons can share enough boundary structure that one replacement ball can participate in repairing both failed cells, reducing the two-fault repair number from four to three.

\begin{figure}[H]
\centering
\includegraphics[width=0.6\linewidth]{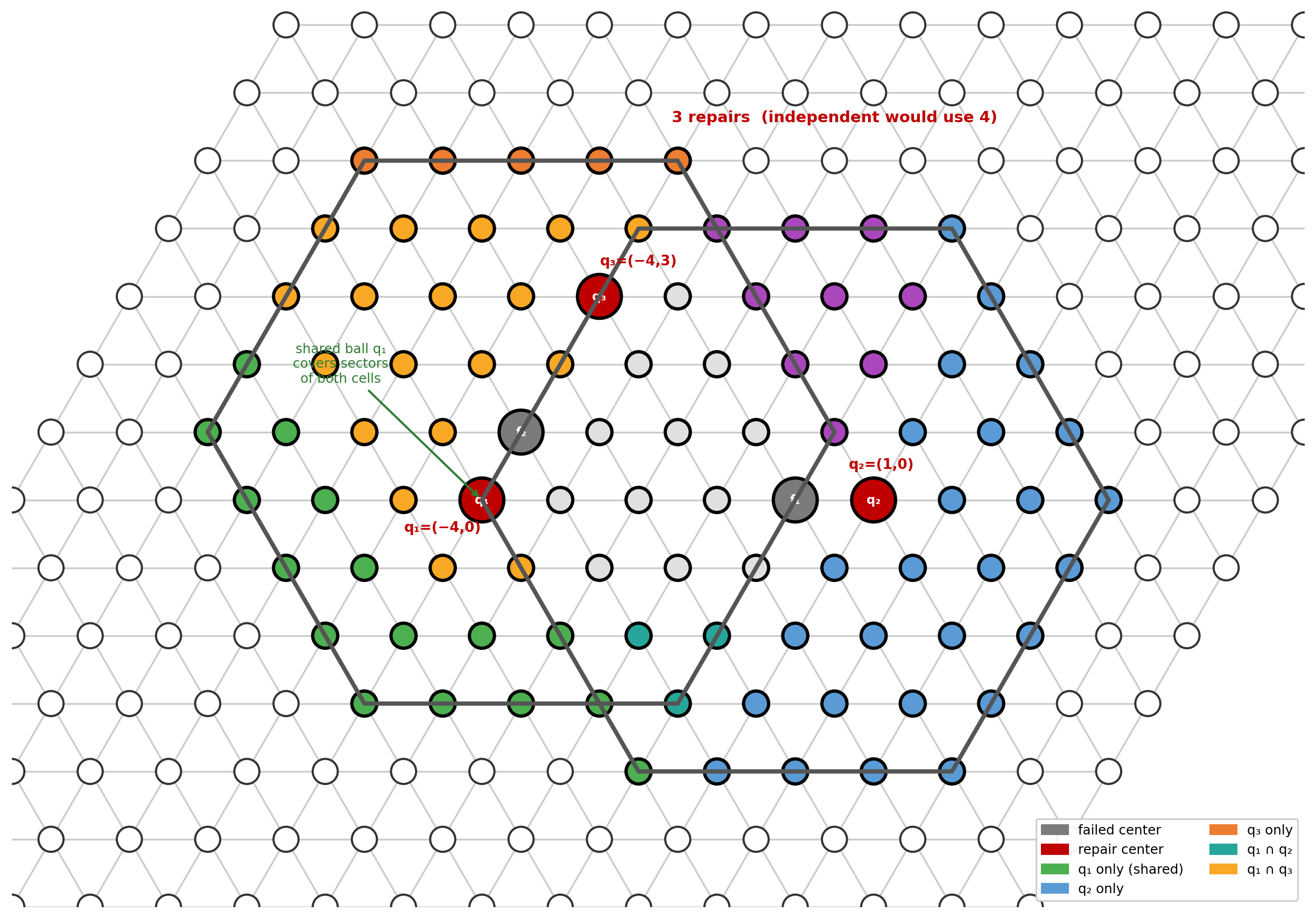}
\caption{Non-additive two-fault repair for the neighboring family $D=g_2=(-t,1)$ at $t=4$.  Three replacement balls cover the union of the two failed hexagons, so the repair is strictly better than independent four-ball repair.}
\label{fig:ej-twofault-three-cover}
\end{figure}

The dense EJ resource lattice associated with $\alpha=(t+1)+t\omega$ is generated in axial coordinates by
\begin{equation}
    g_1=(t+1,t),\qquad g_2=(-t,1).
    \label{eq:ej-resource-generators}
\end{equation}
Thus a neighboring displacement can be written as
\begin{equation}
    D=m g_1+n g_2.
    \label{eq:two-fault-displacement}
\end{equation}
By translation, it suffices to take $r_1=(0,0)$ and $r_2=D$.

\begin{lemma}[First non-additive neighboring family]
Let $t\ge2$ and let the two failed centers be
\begin{equation}
    r_1=(0,0),\qquad r_2=(-t,1)=g_2.
\end{equation}
Then the three centers
\begin{equation}
    R=\{(-t,0),\ (1,0),\ (-t,3)\}
    \label{eq:k3-family-a}
\end{equation}
cover $B_t(r_1)\cup B_t(r_2)$.  By reflection, the displacement $(t,-1)=-g_2$ also has a three-replacement repair.
\end{lemma}

\begin{proof}
A vertex $z=(x,y)$ lies in a ball $B_t(c_1,c_2)$ exactly when
\begin{equation}
    |x-c_1|\le t,\qquad
    |y-c_2|\le t,\qquad
    |x+y-c_1-c_2|\le t.
    \label{eq:three-strip-membership}
\end{equation}
Put
\begin{equation}
 R_1=(-t,0),\qquad R_2=(1,0),\qquad R_3=(-t,3).
\end{equation}
We verify the two failed cells separately by explicit strip intervals.  If $z\in B_t(0,0)$, then $-t\le x,y,x+y\le t$.  If $x\le0$ and $x+y\le0$, then $z\in B_t(R_1)$.  Otherwise $x\ge1$ or $x+y\ge1$.  In the first case, the inequalities $y\ge-t$ and $x+y\ge1-t$ follow from $z\in B_t(0,0)$; in the second case, $x\ge1-t$ follows similarly.  Thus in either case $z$ satisfies
\begin{equation}
1-t\le x\le1+t,\qquad -t\le y\le t,\qquad 1-t\le x+y\le1+t,
\end{equation}
so $z\in B_t(R_2)$.  Hence $B_t(0,0)\subseteq B_t(R_1)\cup B_t(R_2)$.

Now let $z\in B_t(-t,1)$.  Then
\begin{equation}
-2t\le x\le0,\qquad 1-t\le y\le1+t,\qquad 1-2t\le x+y\le1.
\end{equation}
If $y\le t$ and $x+y\le0$, then $z\in B_t(R_1)$.  Otherwise $y=t+1$ or $x+y=1$.  If $y=t+1$, then $x+y\le1$ gives $x\le -t$, and hence $x+y\ge1-t\ge3-2t$ for $t\ge2$.  If $x+y=1$, then $y=1-x\ge1\ge3-t$.  In both cases
\begin{equation}
-2t\le x\le0,\qquad 3-t\le y\le3+t,\qquad 3-2t\le x+y\le3,
\end{equation}
which is exactly the strip system for $B_t(R_3)$.  Therefore $B_t(-t,1)\subseteq B_t(R_1)\cup B_t(R_3)$.  The two containments prove the claimed three-ball cover.  The reflected displacement follows from $(x,y)\mapsto(-x,-y)$.
\end{proof}

\begin{lemma}[Second non-additive neighboring family]
Let $t\ge3$ and let the two failed centers be
\begin{equation}
    r_1=(0,0),\qquad r_2=(1-t,t+2)=g_1+2g_2.
\end{equation}
Then the three centers
\begin{equation}
    R=\{(-t,t+3),\ (0,-1),\ (0,3)\}
    \label{eq:k3-family-b}
\end{equation}
cover $B_t(r_1)\cup B_t(r_2)$.  By reflection, the displacement $(t-1,-t-2)=-(g_1+2g_2)$ also has a three-replacement repair.
\end{lemma}

\begin{proof}
Let
\begin{equation}
 S_1=(-t,t+3),\qquad S_2=(0,-1),\qquad S_3=(0,3).
\end{equation}
Again we use the three-strip membership test.  First let $z=(x,y)\in B_t(0,0)$.  If $y\le t-1$ and $x+y\le t-1$, then $z\in B_t(S_2)$.  Otherwise $y=t$ or $x+y=t$.  If $y=t$, then $x\ge -t$ gives $x+y\ge0\ge3-t$ for $t\ge3$.  If $x+y=t$, then $x\le t$ gives $y\ge0\ge3-t$.  Thus in either case
\begin{equation}
-t\le x\le t,\qquad 3-t\le y\le t+3,\qquad 3-t\le x+y\le t+3,
\end{equation}
so $z\in B_t(S_3)$.  Hence $B_t(0,0)\subseteq B_t(S_2)\cup B_t(S_3)$.

Now let $z\in B_t(1-t,t+2)$.  The failed-cell strips are
\begin{equation}
1-2t\le x\le1,\qquad 2\le y\le2t+2,\qquad 3-t\le x+y\le3+t.
\end{equation}
If $x\le0$ and $y\ge3$, then $z\in B_t(S_1)$, since the $x$- and $y$-ranges and the full $x+y$-range agree with the strips of $B_t(S_1)$ on this part.  Otherwise $x=1$ or $y=2$.  If $x=1$, then $x+y\le3+t$ gives $y\le t+2\le t+3$, and $y\ge2\ge3-t$.  If $y=2$, then $x+y\ge3-t$ gives $x\ge1-t\ge -t$.  In either case $z$ satisfies the three strips of $B_t(S_3)$.  Hence
\begin{equation}
B_t(1-t,t+2)\subseteq B_t(S_1)\cup B_t(S_3).
\end{equation}
Together the two containments prove the cover by the three centers in \eqref{eq:k3-family-b}.  The reflected case follows from $(x,y)\mapsto(-x,-y)$.
\end{proof}

\begin{figure}[H]
\centering
\includegraphics[width=0.6\linewidth]{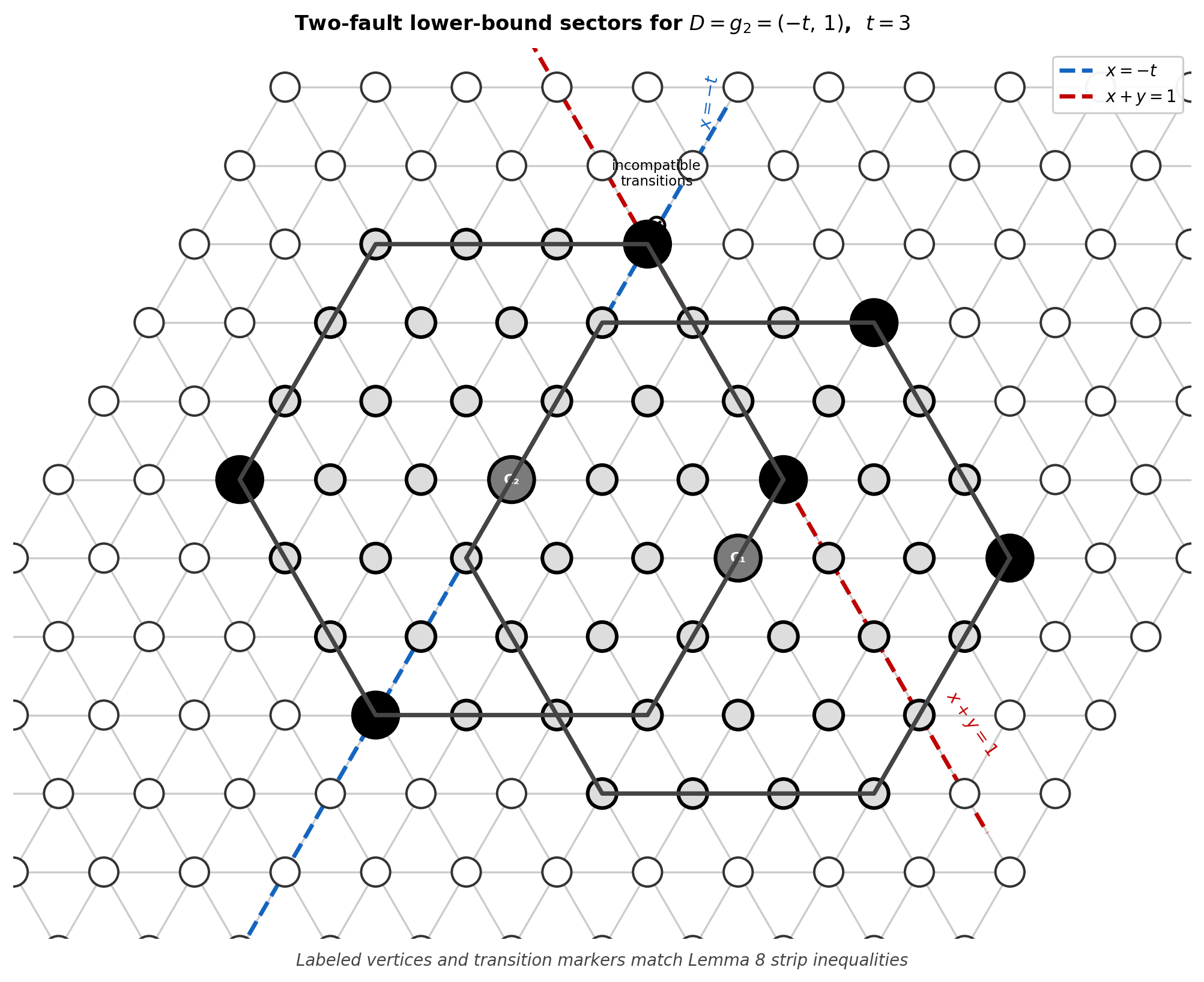}
\caption{Extremal sectors used by the two-fault two-ball obstruction.  The labels mark the axial transition sectors used in Lemma~8; any two-ball cover would have to satisfy incompatible strip inequalities on the two adjacent failed cells.}
\label{fig:ej-twofault-lower-bound}
\end{figure}

Figure~\ref{fig:ej-twofault-lower-bound} marks the incompatible extremal sectors used in the two-ball impossibility argument; the labels in the figure match the transition inequalities used below.

\begin{lemma}[Two-ball impossibility for the non-additive families]
For the two infinite non-additive displacement families $D=\pm g_2$ and $D=\pm(g_1+2g_2)$, no two replacement balls cover $B_t(0)\cup B_t(D)$.  Consequently the three-ball constructions in Lemmas~6 and~7 are optimal.
\end{lemma}

\begin{proof}
It is enough to prove the two positive representatives, because the negative representatives follow by the isometry $(x,y)\mapsto(-x,-y)$.  Suppose first that $D=g_2=(-t,1)$ and that two translated radius-$t$ balls cover $B_t(0)\cup B_t(D)$.  Restricting the cover to $B_t(0)$, the six extreme vertices force the two balls to occupy complementary axial arcs; no available replacement ball can contain an opposite pair of extremes, because such a pair fixes the center at the failed origin.  Thus the two balls determine a single transition strip for $B_t(0)$.  Restricting the same two balls to $B_t(D)$ forces another transition strip.  In the coordinates of \eqref{eq:three-strip-membership}, the first transition requires the shared strip to satisfy $x\ge -t+1$ on the lower endpoint of the adjacent cell, whereas the second requires $x\le -t$ on the corresponding endpoint of the central cell.  These inequalities are incompatible.  Equivalently, adding the two supporting-line inequalities gives $0\le -1$.

For $D=g_1+2g_2=(1-t,t+2)$ we do not invoke an unverified symmetry reduction.  Apply the axial automorphism
\begin{equation}
    T(x,y)=(y,-x-y),
\end{equation}
which permutes the three strip coordinates by
\begin{equation}
    (x,y,x+y)\mapsto (y,-x-y,-x)
\end{equation}
and therefore preserves $d_{\mathrm{EJ}}$.  The image of the second failed center is computed explicitly as
\begin{equation}
    T(1-t,t+2)=(t+2,-3),
    \qquad T(0,0)=(0,0).
\end{equation}
Thus $T$ maps the pair $B_t(0)$ and $B_t(1-t,t+2)$ to the pair $B_t(0)$ and $B_t(t+2,-3)$.  The supporting line $x=-t$ of the original central hexagon is sent by $T$ to the line $X+Y=t$ in the image coordinates $(X,Y)=T(x,y)$, since $X+Y=-x$; this is one of the canonical supporting lines of $B_t(0)$.  Similarly, the adjacent supporting line $x+y=3-t$ of the translated cell is sent to $Y=t-3$, because $Y=-x-y$.  Therefore, in the permuted strip coordinates $(y,-x-y,-x)$, this image has the same adjacent-supporting-strip configuration as the $D=g_2$ case: one transition side requires the common transition coordinate to be at most $-t$, while the other requires the same coordinate to be at least $-t+1$.  Written back in the original coordinates these two inequalities are
\begin{equation}
    x\le -t,\qquad x\ge -t+1,
\end{equation}
for the same transition line.  They cannot hold simultaneously.  Thus no two-ball cover exists for $D=g_1+2g_2$ either, and the reflected representatives follow by symmetry.
\end{proof}

\begin{theorem}[Exact EJ two-fault non-additivity]
\label{thm:nonadditive-twofault}
For the four displacement families
\begin{equation}
    D\in\{g_2,-g_2,g_1+2g_2,-g_1-2g_2\},
\end{equation}
with $t\ge2$ for $\pm g_2$ and $t\ge3$ for $\pm(g_1+2g_2)$,
\begin{equation}
    \rho_{\mathrm{EJ}}^{(2)}(t,\{0,D\})=3.
\end{equation}
Consequently, dense EJ two-fault repair is not universally additive.
\end{theorem}

\begin{proof}
The upper bound three follows from the explicit covers in Lemmas~6 and~7 and their reflected images.  The lower bound three follows from Lemma~8, which rules out any two-ball cover for the same four families.  Thus the exact value is three.  Independent one-fault repair would use four centers, so these families are genuinely non-additive.
\end{proof}

\begin{remark}[Why this differs from the Gaussian case]
The three-replacement families occur because two neighboring EJ hexagons can be cut into three axial sectors that are themselves contained in three translated radius-$t$ hexagons.  The analogous Gaussian failed cells are parity-constrained squares in rotated coordinates; the square-corner obstruction prevents this kind of three-sector cover.  Thus the EJ two-fault result is not a repetition of the Gaussian additivity theorem, but a genuine hexagonal phenomenon.
\end{remark}

\begin{lemma}[Axial endpoint rigidity]
\label{lem:axial-endpoint-rigidity}
Let
\begin{equation}
    e\in\{(1,0),(0,1),(1,-1)\}
\end{equation}
be one of the three axial directions of the EJ hexagon, and let two failed cells have centers $0$ and $D$.  Suppose that
\begin{equation}
\label{eq:endpoint-rigidity-condition}
    d_{\mathrm{EJ}}(D)>2t,\qquad
    d_{\mathrm{EJ}}(D+2te)>2t,\qquad
    d_{\mathrm{EJ}}(D-2te)>2t .
\end{equation}
Then every repair of $B_t(0)\cup B_t(D)$ requires at least four replacement balls.
\end{lemma}

\begin{proof}
Consider the two opposite axial endpoints of the first cell,
\begin{equation}
    p^-=-te,\qquad p^+=te,
\end{equation}
and the corresponding two endpoints of the second cell,
\begin{equation}
    q^-=D-te,
    \qquad q^+=D+te .
\end{equation}
A radius-$t$ EJ ball containing both $p^-$ and $p^+$ must be centered at the midpoint $0$.  To see this algebraically, write $e=(1,0)$; the other two axial directions follow by the coordinate permutations preserving $d_{\mathrm{EJ}}$.  If a center $c=(a,b)$ satisfies
\begin{equation}
    d_{\mathrm{EJ}}(c,(t,0))\le t,
    \qquad
    d_{\mathrm{EJ}}(c,(-t,0))\le t,
\end{equation}
then the two $x$-strip inequalities give
\begin{equation}
    |a-t|\le t,
    \qquad |a+t|\le t,
\end{equation}
which force $a=0$.  With $a=0$, the two $(x+y)$-strip inequalities give
\begin{equation}
    |b-t|\le t,
    \qquad |b+t|\le t,
\end{equation}
which force $b=0$.  Hence the only radius-$t$ ball containing the opposite endpoints of the first failed cell is the failed ball centered at $0$, which is unavailable.  The same argument after translation shows that the only radius-$t$ ball containing both $q^-$ and $q^+$ is the failed ball centered at $D$.

It remains to exclude a single replacement ball covering one endpoint from each failed cell.  The three possible cross differences between an endpoint of the first pair and an endpoint of the second pair are
\begin{equation}
    q^\tau-p^\sigma=D+(\tau-\sigma)te,
    \qquad \sigma,\tau\in\{-1,1\},
\end{equation}
so they are $D$, $D+2te$, or $D-2te$.  By \eqref{eq:endpoint-rigidity-condition}, each has EJ distance strictly larger than $2t$.  Since every radius-$t$ ball has diameter at most $2t$, no replacement ball can contain any such cross pair.  Thus the four endpoint demands must be served by four distinct replacement balls.  Therefore every repair has size at least four.
\end{proof}

\begin{theorem}[Endpoint-rigid additive two-fault pairs]
\label{thm:endpoint-rigid-additivity}
If a two-fault displacement $D$ satisfies the endpoint-rigidity condition \eqref{eq:endpoint-rigidity-condition} for at least one axial direction $e\in\{(1,0),(0,1),(1,-1)\}$, then
\begin{equation}
    \rho_{\mathrm{EJ}}^{(2)}(t,\{0,D\})=4.
\end{equation}
\end{theorem}

\begin{proof}
Lemma~\ref{lem:axial-endpoint-rigidity} gives the lower bound four.  The independent two-cell repair of Theorem~6 gives the upper bound four.  Hence the exact value is four.
\end{proof}

In the bounded neighboring audit, this endpoint-rigidity test covers the additive families whose resource-lattice displacement has no narrow diagonal corridor.  More explicitly, for the tested neighboring displacement families, all stable additive classes except $D=\pm(g_1+g_2)$ satisfy \eqref{eq:endpoint-rigidity-condition} for at least one of the three axial directions; the exceptional diagonal-corridor class is handled separately below.

\begin{lemma}[Diagonal-corridor obstruction]
\label{lem:diagonal-corridor-obstruction}
Let $t\ge3$ and let the two failed centers be
\begin{equation}
    r_1=(0,0),\qquad r_2=g_1+g_2=(1,t+1).
\end{equation}
Then three radius-$t$ replacement balls cannot cover
\begin{equation}
    B_t(0,0)\cup B_t(1,t+1).
\end{equation}
The same conclusion holds for the reflected displacement $-(g_1+g_2)=(-1,-t-1)$.
\end{lemma}

\begin{proof}
We prove the positive displacement; the reflected case follows from the isometry $(x,y)\mapsto(-x,-y)$.  Put
\begin{equation}
    D=(1,t+1),\qquad U=B_t(0,0)\cup B_t(D).
\end{equation}
The proof uses a finite boundary-demand set, but the argument is purely algebraic: each demand point is a closed-form function of $t$, and all exclusions are verified by the three strip inequalities defining $d_{\rm EJ}$.

Consider the ten boundary points
\begin{align}
\mathcal P_0={}&\{(-t,0),(-t,t),(0,-t),(t,-t),(t,0)\},\label{eq:diag-demand-P0}\\
\mathcal P_D={}&\{(1-t,t+1),(1-t,2t+1),(1,2t+1),(t+1,1),(t+1,t+1)\},\label{eq:diag-demand-PD}
\end{align}
and let $\mathcal P=\mathcal P_0\cup\mathcal P_D$.  The first five points lie on the boundary of $B_t(0,0)$, and the second five are the corresponding five boundary points of $B_t(D)$.  The omitted extreme points, $(0,t)$ and $(1,1)$, are exactly the two extremes facing the narrow diagonal corridor.  Thus $\mathcal P\subseteq U$.

We claim that every nonfailed radius-$t$ ball contains at most three points of $\mathcal P$.  Let $c=(a,b)$ be a center with $c\ne(0,0)$ and $c\ne D$.  If $B_t(c)$ contains four points of $\mathcal P_0$, then the three strip coordinates of $c$ are forced to be zero.  Indeed, any four of the five points in \eqref{eq:diag-demand-P0} include two points on one pair of opposite supporting sides and one point on a second pair; applying the inequalities
\begin{equation}
    |x-a|\le t,
    \qquad |y-b|\le t,
    \qquad |x+y-a-b|\le t
    \label{eq:diag-basic-strips}
\end{equation}
to those boundary points gives successively $a=0$, $b=0$, and $a+b=0$.  Hence the only ball containing four of $\mathcal P_0$ is the failed ball $B_t(0,0)$, which is unavailable.  Translating the same calculation by $D$ shows that the only ball containing four of $\mathcal P_D$ is the failed ball $B_t(D)$, also unavailable.

It remains to rule out mixed four-point coverage.  A mixed four-point subset contains at least two points from one of the two five-point sets and at least two points from the other.  Up to reversing the two failed cells and reflecting across the diagonal corridor, the possible two-point boundary types are represented by the following three strip separations:
\begin{equation}
\begin{array}{c|c}
\text{boundary type} & \text{forced separation between the two pairs}\\ \hline
x\text{-parallel pair} & |(t+1)-(-t)|=2t+1,\\[1mm]
y\text{-parallel pair} & |(2t+1)-(-t)|=3t+1,\\[1mm]
(x+y)\text{-parallel pair} & |(2t+2)-0|=2t+2.
\end{array}
\label{eq:diag-mixed-separations}
\end{equation}
Each displayed separation is strictly larger than $2t$.  Since any radius-$t$ EJ ball has diameter at most $2t$ in every strip coordinate, no available radius-$t$ ball can contain two boundary points of one failed cell and two boundary points of the other.  Therefore an available radius-$t$ ball contains at most three points of $\mathcal P$.

If three replacement balls covered $U$, then they would cover all ten points of $\mathcal P$.  But each available replacement ball covers at most three of these points, so three balls cover at most nine points of $\mathcal P$, a contradiction.  Hence at least four replacement balls are necessary for $D=g_1+g_2$.  Reflection gives the same lower bound for $D=-(g_1+g_2)$.
\end{proof}

\begin{theorem}[Diagonal-corridor additivity]
\label{thm:diagonal-corridor-additivity}
For $t\ge3$ and
\begin{equation}
    D=\pm(g_1+g_2)=\pm(1,t+1),
\end{equation}
we have
\begin{equation}
    \rho_{\mathrm{EJ}}^{(2)}(t,\{0,D\})=4.
\end{equation}
\end{theorem}

\begin{proof}
The lower bound four follows from Lemma~\ref{lem:diagonal-corridor-obstruction}.  The upper bound four follows from independent two-cell repair.  Hence the exact value is four.
\end{proof}

\begin{remark}[Scope of the two-fault classification]
The two-fault section uses four symbolic mechanisms: explicit three-ball covers for the stable non-additive families, two-ball impossibility for those same families, endpoint rigidity for additive pairs with separated opposite axial endpoints, and the diagonal-corridor obstruction for $D=\pm(g_1+g_2)$.  These theorems do not assert that every neighboring displacement is captured by a single closed formula; the bounded audit records the remaining tested neighboring cases and the small-radius exceptions.
\end{remark}

\begin{lemma}[No shared candidate for distant failed cells]
\label{lem:no-shared-distant-cells}
Let $r_i$ and $r_j$ be two failed resource centers.  If
\begin{equation}
    d_{\mathrm{EJ}}(r_i,r_j)>4t,
\end{equation}
then no radius-$t$ replacement ball can intersect both $B_t(r_i)$ and $B_t(r_j)$.
\end{lemma}

\begin{proof}
Suppose that a replacement center $z$ has
\begin{equation}
    B_t(z)\cap B_t(r_i)\ne\emptyset,
    \qquad
    B_t(z)\cap B_t(r_j)\ne\emptyset.
\end{equation}
Choose vertices $u\in B_t(z)\cap B_t(r_i)$ and $v\in B_t(z)\cap B_t(r_j)$.  Then
\begin{align}
    d_{\mathrm{EJ}}(r_i,r_j)
    &\le d_{\mathrm{EJ}}(r_i,u)+d_{\mathrm{EJ}}(u,z)
       +d_{\mathrm{EJ}}(z,v)+d_{\mathrm{EJ}}(v,r_j) \\
    &\le 4t,
\end{align}
contradicting the hypothesis.
\end{proof}

\begin{corollary}[Separated two-fault repair]
If two failed resources $r_1,r_2$ satisfy $d_{\mathrm{EJ}}(r_1,r_2)>4t$, then
\begin{equation}
    \rho_{\mathrm{EJ}}^{(2)}(t,\{r_1,r_2\})=4.
\end{equation}
\end{corollary}

\begin{proof}
The upper bound is Theorem~6.  For the lower bound, Lemma~\ref{lem:no-shared-distant-cells} shows that every replacement ball can intersect at most one failed cell.  Each failed cell requires at least two replacements by Theorem~3.  Hence the two failed cells require at least four replacement balls.
\end{proof}

\section{Multi-Fault Repair Bounds and Dense-Cluster Subadditivity}

Let
\begin{equation}
    F=\{r_1,r_2,\ldots,r_q\}\subseteq S
\end{equation}
be the set of failed resources and let
\begin{equation}
    U(F)=\bigcup_{r\in F} B_t(r)
    \label{eq:failed-region}
\end{equation}
be the failed region.  A repair set $R$ is valid if
\begin{equation}
    U(F)\subseteq \bigcup_{z\in R}B_t(z).
\end{equation}
For a fixed failed set $F$, write $\rho_{\mathrm{EJ}}^{(q)}(t,F)$ for the minimum number of replacements needed to cover $U(F)$.

\begin{theorem}[Independent $q$-fault upper bound]
For every failed set $F$ with $|F|=q$,
\begin{equation}
    \rho_{\mathrm{EJ}}^{(q)}(t,F)\le 2q.
\end{equation}
\end{theorem}

\begin{proof}
For each failed resource $r\in F$, choose one translated canonical repair pair from \eqref{eq:translated-algorithm}, for example
\begin{equation}
    R_r=\{r-(1,0),\ r+t(1,0)\}.
\end{equation}
By Corollary~1, $R_r$ covers $B_t(r)$.  Therefore
\begin{equation}
    R=\bigcup_{r\in F}R_r
\end{equation}
covers $\bigcup_{r\in F}B_t(r)=U(F)$.  Since $|R|\le\sum_{r\in F}|R_r|=2q$, the claimed upper bound follows.
\end{proof}

\begin{theorem}[Exact additivity for separated failures]
Let $F=\{r_1,\ldots,r_q\}$ be a failed resource set satisfying
\begin{equation}
    d_{\mathrm{EJ}}(r_i,r_j)>4t
    \qquad\text{for all }i\ne j.
\end{equation}
Then
\begin{equation}
    \rho_{\mathrm{EJ}}^{(q)}(t,F)=2q.
\end{equation}
\end{theorem}

\begin{proof}
The upper bound is the independent repair theorem.  For the lower bound, Lemma~\ref{lem:no-shared-distant-cells} shows that every replacement ball can intersect at most one failed cell.  Each individual failed cell requires at least two replacements by Theorem~3.  Hence the $q$ failed cells require at least $2q$ replacement balls in total.  The lower and upper bounds agree.
\end{proof}

The separated theorem describes the additive regime.  The clustered regime is more interesting.  
The clustered regime also contains small closed-form families whose optima can be proved directly.  The next theorem gives a four-fault algebraic subadditivity result before the six-fault saving theorem and shows that the phenomenon is not confined to a single large example.

\begin{theorem}[Exact dense four-fault cluster]
Let
\begin{equation}
\begin{split}
C_t=\{&(-t,1),\ (-1,-t-1),\ (0,0),\ (t-1,-t-2)\}.
\end{split}
\label{eq:four-fault-cluster}
\end{equation}
For every $t\ge3$,
\begin{equation}
    \rho_{\mathrm{EJ}}^{(4)}(t,C_t)=4.
\end{equation}
Thus four dense failed resource cells can be repaired with four replacements rather than the independent bound $2q=8$.
\end{theorem}

\begin{proof}
First consider the four replacement centers
\begin{equation}
S_t=\{(-t-1,2),\ (-1,-t-2),\ (1,0),\ (t,-t-2)\}.
\label{eq:four-fault-repair}
\end{equation}
Using the three-strip representation of an EJ ball, direct interval substitution gives the following containment relations:
\begin{align}
B_t(-t,1)&\subseteq B_t(S_1)\cup B_t(S_2)\cup B_t(S_3),\nonumber\\
B_t(-1,-t-1)&\subseteq B_t(S_1)\cup B_t(S_2)\cup B_t(S_4),\nonumber\\
B_t(0,0)&\subseteq B_t(S_1)\cup B_t(S_2)\cup B_t(S_3),\nonumber\\
B_t(t-1,-t-2)&\subseteq B_t(S_2)\cup B_t(S_4),
\label{eq:four-cover-audit}
\end{align}
where $S_1,\ldots,S_4$ are listed in the order of \eqref{eq:four-fault-repair}.  For example, a point in $B_t(t-1,-t-2)$ that is not contained in $B_t(S_2)$ violates exactly the left axial strip relative to $S_2$; the same strip inequality places it in $B_t(S_4)$.  For the first containment in \eqref{eq:four-cover-audit}, the left extreme point $p=(-2t,1)$ of $B_t(-t,1)$ is covered by $S_1=(-t-1,2)$ because
\begin{equation}
 d_{\rm EJ}(p,S_1)=\max\{t-1,1,t\}=t.
\end{equation}
The opposite transition boundary of this same failed cell is handled by the $S_2$ and $S_3$ strips: substituting a boundary point $(-t+u,1+v)$ with $|u|,|v|,|u+v|\le t$ shows that failure of the $S_1$ right strip leaves the point in the $S_3$ interval, while failure of the $S_1$ lower diagonal strip leaves it in the $S_2$ interval.  These are the same two shifted inequalities used in the displayed containment, and they cover the remaining boundary arcs.  The other containments are analogous three-strip interval calculations.  Hence $S_t$ covers the four failed cells and $\rho_{\mathrm{EJ}}^{(4)}(t,C_t)\le4$.

For the lower bound, define the four-point demand packing
\begin{equation}
Z_t=\{(-2t,1),\ (-1,-2t-1),\ (0,t),\ (2t-1,-t-2)\}.
\label{eq:four-demand-set}
\end{equation}
These points lie in the four failed cells of \eqref{eq:four-fault-cluster}, respectively.  Their pairwise EJ distances are
\begin{equation}
\begin{array}{c|ccc}
 & z_2&z_3&z_4\\
\hline
z_1&2t+2&3t-1&4t-1\\
z_2&     &3t+2&3t-1\\
z_3&     &    &2t+2
\end{array}
\end{equation}
all of which are strictly larger than $2t$ for $t\ge3$.  Since every radius-$t$ EJ ball has diameter at most $2t$, one replacement ball can meet at most one point of this strict $2t$-packing.  Four distinct replacement balls are necessary, so $\rho_{\mathrm{EJ}}^{(4)}(t,C_t)\ge4$, proving the theorem.
\end{proof}

The exact optimizer repeatedly returned a six-fault pattern whose coordinates are independent of $t$ when written in the resource-lattice basis $g_1,g_2$.  In axial coordinates the family is
\begin{equation}
\begin{split}
F_t=\{&(-t,1),\ (-1,-t-1),\ (0,0),\\
      &(t-1,-t-2),\ (t,-1),\ (2t,-2)\}.
\end{split}
\label{eq:six-fault-cluster}
\end{equation}
Equivalently, in $(m,n)$ resource-lattice coordinates with $D=mg_1+ng_2$, the same pattern is
\begin{equation}
    \{(0,1),\ (-1,-1),\ (0,0),\ (-1,-2),\ (0,-1),\ (0,-2)\}.
\end{equation}
Thus the shape is a fixed dense two-column cluster in the resource lattice.  Figure~\ref{fig:ej-six-fault-cluster} illustrates this cluster, the five replacement balls, and the packing points for $t=3$.

\begin{figure}[H]
\centering
\includegraphics[width=0.65\linewidth]{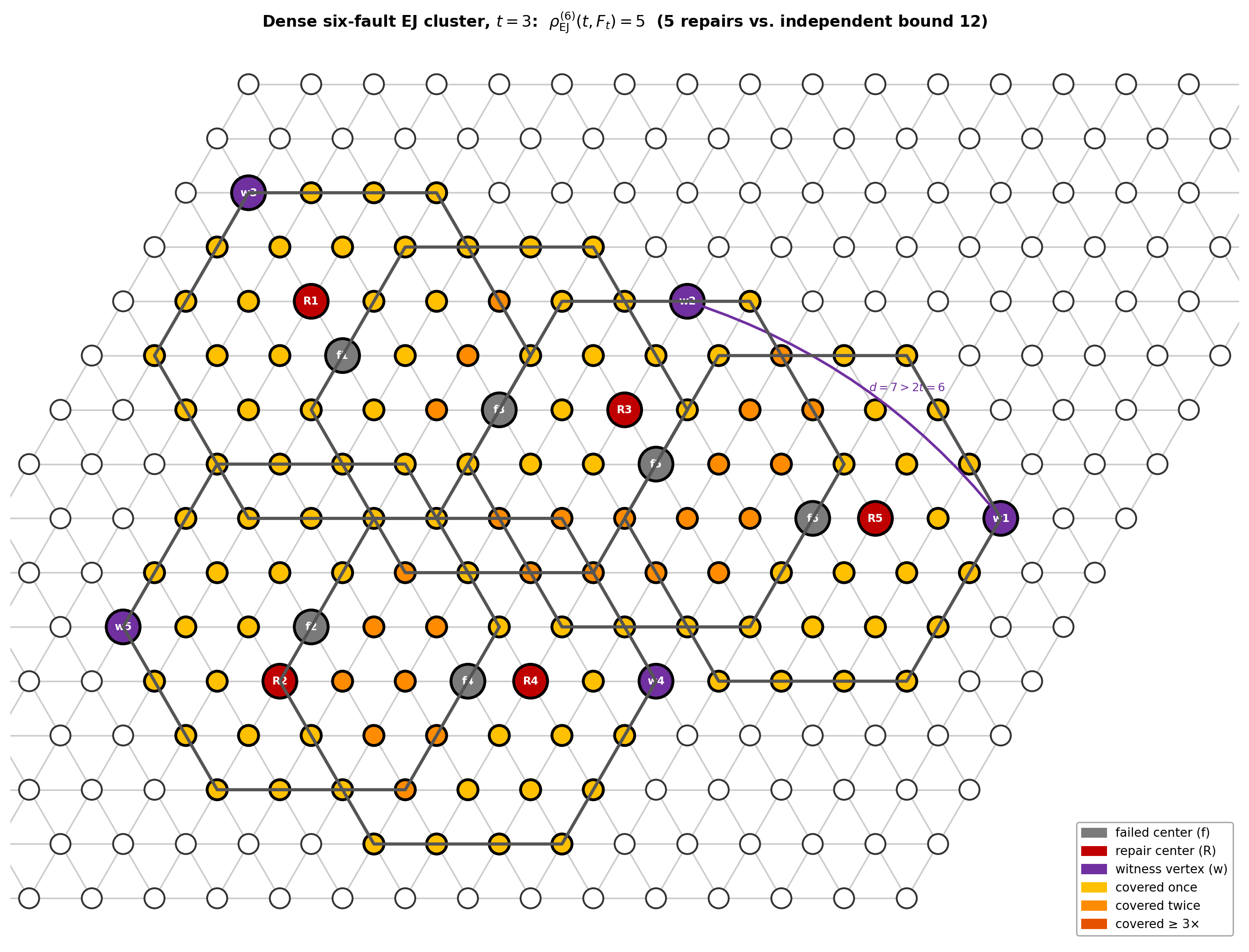}
\caption{Dense six-fault EJ cluster for $t=3$.  The six failed resource cells in $F_t$ are repaired by the five centers in $R_t$, while the marked packing points are pairwise farther than $2t$ and therefore give the demand-packing lower bound.  This is the proof-supporting visualization for the exact subadditivity theorem $\rho_{\mathrm{EJ}}^{(6)}(t,F_t)=5$.}
\label{fig:ej-six-fault-cluster}
\end{figure}

\begin{lemma}[Five-ball cover of the six-fault cluster]
For every $t\ge3$, the failed region generated by $F_t$ in \eqref{eq:six-fault-cluster} is covered by the five replacement centers
\begin{equation}
\begin{split}
R_t=\{&(-t-1,2),\ (-1,-t-2),\ (2,t-3),\\
      &(t,-t-2),\ (2t+1,-2)\}.
\end{split}
\label{eq:six-fault-five-repair}
\end{equation}
Therefore $\rho_{\mathrm{EJ}}^{(6)}(t,F_t)\le5$.
\end{lemma}

\begin{proof}
Membership in an EJ ball is the three-strip condition
\begin{equation}
    |x-c_1|\le t,\qquad |y-c_2|\le t,
    \qquad |x+y-c_1-c_2|\le t.
    \label{eq:strip-membership-multi}
\end{equation}
Substituting the six failed centers in \eqref{eq:six-fault-cluster} and the five repair centers in \eqref{eq:six-fault-five-repair} gives the following strip-containment relations:
\begin{align}
B_t(-t,1)&\subseteq B_t(R_1)\cup B_t(R_2)\cup B_t(R_3)\cup B_t(R_4),\nonumber\\
B_t(-1,-t-1)&\subseteq B_t(R_1)\cup B_t(R_2)\cup B_t(R_4),\nonumber\\
B_t(0,0)&\subseteq B_t(R_1)\cup B_t(R_2)\cup B_t(R_3)\cup B_t(R_4),\nonumber\\
B_t(t-1,-t-2)&\subseteq B_t(R_2)\cup B_t(R_4),\nonumber\\
B_t(t,-1)&\subseteq B_t(R_3)\cup B_t(R_4)\cup B_t(R_5),\nonumber\\
B_t(2t,-2)&\subseteq B_t(R_3)\cup B_t(R_4)\cup B_t(R_5),
\label{eq:six-cover-audit}
\end{align}
where $R_1,\ldots,R_5$ denote the five centers in the order listed in \eqref{eq:six-fault-five-repair}.  Each line of \eqref{eq:six-cover-audit} is obtained by expanding \eqref{eq:strip-membership-multi}.  For example, the failed cell $B_t(t-1,-t-2)$ is shared by $R_2=(-1,-t-2)$ and $R_4=(t,-t-2)$.  A concrete boundary trace is the vertex
\begin{equation}
    p=(2t-1,-t-2)\in B_t(t-1,-t-2).
\end{equation}
It is not in $B_t(R_2)$ because $|p_x+1|=2t>t$, but it is in $B_t(R_4)$ since
\begin{equation}
    |p_x-t|=t-1,
    \qquad |p_y+t+2|=0,
    \qquad |p_x+p_y+2|=t-1.
\end{equation}
Thus the same three strip inequalities that exclude $p$ from the left repair ball place it inside the right repair ball.  The symbolic interval calculation in \eqref{eq:six-cover-audit} applies this interval check to every vertex of each failed cell.  Taking the union of the six containments proves that $R_t$ covers the whole failed region.
\end{proof}

\begin{lemma}[Five-point demand packing]
For every $t\ge3$, the failed region $U(F_t)$ contains a five-point demand packing
\begin{equation}
\begin{split}
W_t=\{&(3t,-2),\ (t-1,t-1),\ (-2t,t+1),\\
      &(2t-1,-t-2),\ (-t-1,-t-1)\}
\end{split}
\label{eq:five-packing-points}
\end{equation}
whose pairwise distances are all greater than $2t$.
\end{lemma}

\begin{proof}
Each demand point lies in the failed region.  Specifically,
\begin{align}
(3t,-2)&\in B_t(2t,-2),&
(t-1,t-1)&\in B_t(t,-1),\nonumber\\
(-2t,t+1)&\in B_t(-t,1),&
(2t-1,-t-2)&\in B_t(t-1,-t-2),\nonumber\\
(-t-1,-t-1)&\in B_t(-1,-t-1).\nonumber
\end{align}
For any two vertices $u,v$, a radius-$t$ ball can cover both only if $d_{\mathrm{EJ}}(u,v)\le2t$.  The pairwise distances among the five demand points are
\begin{equation}
\begin{array}{c|cccc}
 & w_2&w_3&w_4&w_5\\
\hline
w_1&2t+1&5t&2t+1&5t\\
w_2&     &3t-1&2t+1&4t\\
w_3&     &    &4t-1&2t+2\\
w_4&     &    &    &3t
\end{array}
\end{equation}
where $w_1,\ldots,w_5$ are ordered as in \eqref{eq:five-packing-points}.  For $t\ge3$ every displayed distance is strictly larger than $2t$; the binding entry is $3t-1$, which already exceeds $2t$ for every $t>1$.  The construction also covers the boundary case $t=3$, where the third repair center in \eqref{eq:six-fault-five-repair} becomes $(2,0)$ and all displayed strip inequalities remain valid.  Since the diameter of a radius-$t$ EJ ball is at most $2t$, any single replacement ball can cover at most one point of this packing.  Therefore at least five replacement balls are necessary.
\end{proof}

\begin{theorem}[Exact dense six-fault subadditivity]
For the six-fault cluster $F_t$ in \eqref{eq:six-fault-cluster}, and every $t\ge3$,
\begin{equation}
    \rho_{\mathrm{EJ}}^{(6)}(t,F_t)=5.
\end{equation}
Consequently, the independent bound $2q=12$ can overestimate the optimum by seven replacements.
\end{theorem}

\begin{proof}
The five-ball construction above gives $\rho_{\mathrm{EJ}}^{(6)}(t,F_t)\le5$.  The five-point demand-packing lemma gives the reverse inequality.  Therefore the exact value is five.
\end{proof}

The six-fault theorem is included not as a numerical curiosity but as a structural result: dense EJ clusters can share repair balls across several adjacent failed service cells.  This behavior is impossible in the separated regime and is much stronger than the two-fault drop from four to three.

\begin{remark}[Role of the $t\le12$ multi-fault audit]
The dense four-fault and six-fault cluster theorems above are symbolic results valid for all $t\ge3$; they do not depend on extrapolating the finite audit.  The audit is used in the same role as in the Gaussian local-repair paper: it validates the tested finite search space, validates the implementation of the coverage and overlap identities, and exposes additional patterns for future closed-form classification.  The main mathematical claims in this section are the universal $2q$ upper bound, separated additivity, the exact four-fault cluster, the exact six-fault cluster, and the inclusion--exclusion overlap identity.
\end{remark}

\section{Exact Multi-Fault Overlap Accounting}

For multiple failures, the replacement number is not the only useful quantity.  Independent local repairs may interact, creating vertices of the failed region that are covered two, three, or more times by replacement balls.  The following identity is unconditional: it applies to any valid repair set in any failed region.

For $z\in U(F)$, define the replacement multiplicity
\begin{equation}
    \mu_R(z)=|\{r\in R:z\in B_t(r)\}|.
\end{equation}
The true overlap mass inside the failed region is
\begin{equation}
    O(R)=\sum_{z\in U(F)}(\mu_R(z)-1).
    \label{eq:true-overlap-mass}
\end{equation}
For $j\ge2$, define
\begin{equation}
    P_j(R)=\sum_{z\in U(F)}\binom{\mu_R(z)}{j}.
\end{equation}
Thus $P_2(R)$ is the total pairwise common-intersection mass, $P_3(R)$ is the total triple-intersection mass, and so on.

\begin{lemma}[Vertexwise binomial identity]
For every integer $m\ge1$,
\begin{equation}
    m-1=\sum_{j=2}^{m}(-1)^j\binom{m}{j}.
\end{equation}
\end{lemma}

\begin{proof}
The binomial identity $\sum_{j=0}^{m}(-1)^j\binom{m}{j}=0$ gives
\begin{equation}
    1-m+\sum_{j=2}^{m}(-1)^j\binom{m}{j}=0,
\end{equation}
which is equivalent to the stated formula.
\end{proof}

\begin{theorem}[Exact multi-failure overlap identity]
For any valid repair $R$ of any failed resource set $F$,
\begin{equation}
    O(R)=\sum_{j=2}^{M}(-1)^jP_j(R),
    \label{eq:multifault-ie}
\end{equation}
where
\begin{equation}
    M=\max_{z\in U(F)}\mu_R(z).
\end{equation}
\end{theorem}

\begin{proof}
Apply the vertexwise identity to $m=\mu_R(z)$ at each $z\in U(F)$ and sum over all failed-region vertices.  Exchanging the finite sums gives exactly \eqref{eq:multifault-ie}.
\end{proof}

\begin{corollary}[Triple-core specialization]
If a valid repair set $R$ satisfies $\max_{z\in U(F)}\mu_R(z)\le3$, then
\begin{equation}
    O(R)=P_2(R)-P_3(R).
\end{equation}
Moreover, if $C_3(R)$ is the number of failed-region vertices covered by exactly three replacement balls, then $P_3(R)=C_3(R)$ and
\begin{equation}
    O(R)=P_2(R)-C_3(R).
\end{equation}
\end{corollary}

\begin{proof}
When the maximum multiplicity is at most three, all terms $P_j$ with $j\ge4$ vanish.  Also $\binom{\mu_R(z)}{3}$ is one exactly at vertices with multiplicity three and zero elsewhere.
\end{proof}

For an independent canonical repair of $q$ failed EJ cells, the additive one-cell overlap prediction is
\begin{equation}
    A_q=q t^2.
\end{equation}
The extra dense-core overlap is
\begin{equation}
    \Omega_{\mathrm{extra}}(R)=O(R)-A_q.
\end{equation}
Combining this definition with the exact multi-failure overlap identity gives the general correction formula
\begin{equation}
    \Omega_{\mathrm{extra}}(R)=\sum_{j=2}^{M}(-1)^jP_j(R)-qt^2.
\end{equation}
If a repair has $M\le3$, this reduces to
\begin{equation}
    \Omega_{\mathrm{extra}}(R)=P_2(R)-C_3(R)-qt^2.
\end{equation}

\begin{example}[Multiplicity-four overlap accounting]
The full formula is needed in EJ dense clusters.  One certified audit row has $t=4$, three failed resources
\begin{equation}
    F=\{(-4,1),\ (-1,-5),\ (0,0)\},
\end{equation}
and four replacement centers
\begin{equation}
    R=\{(-5,2),\ (-1,-6),\ (1,-1),\ (1,0)\}.
\end{equation}
The computed multiplicity profile has maximum multiplicity $M=4$ and
\begin{equation}
    P_2=95,\qquad P_3=20,\qquad P_4=1.
\end{equation}
Therefore the true repeated-coverage mass is
\begin{equation}
    O(R)=P_2-P_3+P_4=95-20+1=76.
\end{equation}
This example also shows why the triple-core shortcut is not universal.  The same row has sixteen vertices of exact multiplicity three and one vertex of multiplicity four; the multiplicity-four vertex contributes $\binom{4}{3}=4$ to $P_3$ and $\binom{4}{4}=1$ to $P_4$.  Thus replacing $P_3$ by only the number of exactly triple-covered vertices would give the wrong overlap.  The additive one-fault prediction for this row is $3t^2=48$, so the dense-core excess is $76-48=28$.
\end{example}

\section{Exact Validation and Audit}
\label{sec:validation}

The theorems above are geometric and do not rely on simulation, but the implementation was audited in the same spirit as the Gaussian local-repair study.  The one-fault audit exhaustively enumerated all candidates in $B_{2t}(0)\setminus\{0\}$ for $1\le t\le20$, minimized repair count and then overlap, and confirmed $\rho_{\mathrm{EJ}}(t)=2$, $\Omega_{\mathrm{EJ}}(t)=t^2$, and the $3t$ canonical optimum-pair count for every tested radius.  The targeted two-fault audit enumerated all neighboring displacements $D=mg_1+ng_2$ with $d_{\mathrm{EJ}}(0,D)\le4t$ for $2\le t\le20$ and solved 832 exact set-cover instances: 78 have $K=3$, 754 have $K=4$, and no row is unresolved.  The stable $K=3$ rows are the infinite families $D=\pm g_2$ and $D=\pm(g_1+2g_2)$ proved above; the additional $K=3$ rows occur only at $t=2$ and $t=4$ and are recorded as finite short-corridor exceptions.

The multi-fault optimizer solved 19,400 exact instances for $2\le t\le12$ and $3\le q\le6$ in clustered, line, pair-plus, triangle, and random-local modes.  The optimum replacement counts show strong subadditivity: maximum savings relative to independent repair were 2, 4, 5, and 7 for $q=3,4,5,6$, respectively, and maximum replacement multiplicity was 4.  The smooth increase in average saving across $q$ reflects the balanced audit design and the growing number of dense local sharing opportunities, not a proved asymptotic law.  Every audit row records failed centers, replacement centers, uncovered vertices, multiplicity counts, higher-order coverage masses, and inclusion--exclusion residuals; the residual is zero in all rows.  The audit is therefore a reproducibility and pattern-discovery layer, while the universal mathematical claims are the one-fault formulas, the symbolic two-fault mechanisms, separated additivity, the dense four- and six-fault cluster theorems, and the overlap identity.

\section{Conclusion}

This paper introduced local fault repair for perfect resource placements in dense Eisenstein--Jacobi networks.  After one resource fails, the uncovered region is exactly the former hexagonal service cell, and every useful replacement lies within distance $2t$ of the failed resource.  A single nonfailed replacement cannot cover the failed cell, while two explicitly placed replacements always do.  Hence $\rho_{\mathrm{EJ}}(t)=2$ for every $t\ge1$.  Among all minimum-size repairs, every two-ball cover has at least $t^2$ repeated vertices inside the failed cell, and the canonical axial repair attains this bound exactly.  Thus $\Omega_{\mathrm{EJ}}(t)=t^2$.

For two failed resources, EJ repair is not universally additive.  Four replacements always suffice by independent repair, but two infinite neighboring displacement families admit three-replacement repairs.  Additive behavior is proved algebraically for endpoint-rigid pairs, for the diagonal-corridor pair $D=\pm(g_1+g_2)$, and for separated failed cells.  A targeted exact optimizer for $2\le t\le20$ classified 832 neighboring displacement cases, finding 78 non-additive $K=3$ cases and 754 additive $K=4$ cases with no unresolved rows, in agreement with the symbolic mechanisms and the listed small-radius exceptions.  This separates the EJ theory from the Gaussian local-repair paper, where the two-fault value is additive for all displacements.

For multiple failed resources, the paper proves both an additive regime and a dense-cluster subadditive regime.  Independent canonical repair gives a $2q$ upper bound for $q$ failed resources, and this bound is exact whenever failed cells are pairwise more than $4t$ apart.  In contrast, a fixed four-fault dense cluster has exact repair number four and a fixed six-fault dense cluster has exact repair number five for every $t\ge3$, giving savings of four and seven replacements relative to independent repair.  The exact multi-fault audit for $2\le t\le12$ confirms that this subadditivity is common in clustered EJ failures.  Finally, repeated coverage inside any failed region is governed by an exact inclusion--exclusion identity over replacement multiplicities.  Together these results show that local repair in Eisenstein--Jacobi placements has a distinct hexagonal theory: one-fault overlap is quadratic, two-fault repair can be non-additive, and dense multi-fault clusters can reuse replacement balls across several failed service cells.

\section*{Acknowledgments}
The author thanks the Department of Computer Science, Faculty of Science, Kuwait University, for its support and research environment.  This work did not receive a specific grant from any funding agency in the public, commercial, or not-for-profit sectors.

\section*{Declaration of competing interest}
The author declares that he has no known competing financial interests or personal relationships that could have appeared to influence the work reported in this paper.

\section*{Data availability}
The validation data and audit data generated during the current study are available from the author upon reasonable request.

\section*{Declaration of generative AI and AI-assisted technologies in the writing process}
During the preparation of this work, the author used an AI-assisted language tool, to support manuscript wording refinement, and LaTeX preparation. The author reviewed and edited all AI-assisted output and takes full responsibility for the content of the published article.

\end{document}